\documentclass{article}

\usepackage{arxiv}

\usepackage[utf8]{inputenc} 
\usepackage[T1]{fontenc}    
\usepackage[hidelinks]{hyperref}       
\usepackage{url}            
\usepackage{booktabs}       
\usepackage{amsfonts}       
\usepackage{nicefrac}       
\usepackage{microtype}      
\usepackage{lipsum}		
\usepackage{graphicx}
\usepackage{natbib}
\usepackage{doi}
\usepackage{amsmath}
\usepackage[inkscapearea=page]{svg}
\usepackage[ruled]{algorithm2e}
\usepackage{amsthm}
\usepackage{dsfont}
\usepackage{comment}

\theoremstyle{plain}
\newtheorem{thm}{\protect\theoremname}
\theoremstyle{plain}
\newtheorem{prop}[thm]{\protect\propositionname}
\theoremstyle{definition}

\theoremstyle{definition}
\newtheorem{ass}{\protect\assumptionname}[section]
\theoremstyle{definition}
\newtheorem*{rem}{\protect\remarkname}
\theoremstyle{plain}

\theoremstyle{plain}
\newtheorem{cor}[thm]{\protect\corollaryname}
\newtheorem{result}{Result}

\providecommand{\definitionname}{Definition}
\providecommand{\propositionname}{Proposition}
\providecommand{\theoremname}{Theorem}
\providecommand{\assumptionname}{Assumption}
\providecommand{\remarkname}{Remark}
\providecommand{\lemmaname}{Lemma}
\providecommand{\corollaryname}{Corollary}

\title{Quantifying the effectiveness of linear preconditioning in Markov chain Monte Carlo}

\author{ \href{https://maxhhird.github.io/}{Max Hird} \\
	Department of Statistical Science\\
	University College London\\
	1–19 Torrington Place, London WC1E 7HB \\
	\texttt{max.hird.19@ucl.ac.uk} \\
	\And
	\href{https://samueljlivingstone.wixsite.com/webpage}{Samuel Livingstone} \\
	Department of Statistical Science\\
	University College London\\
	1–19 Torrington Place, London WC1E 7HB \\
}

\date{}

\renewcommand{\headeright}{}
\renewcommand{\undertitle}{}
\renewcommand{\shorttitle}{Quantifying the effectiveness of linear preconditioning in Markov chain Monte Carlo}

\hypersetup{
pdftitle={Quantifying the effectiveness of linear preconditioning in Markov chain Monte Carlo},
pdfsubject={stat.CO, stat.TH, stat.OT, math.ST},
pdfauthor={Max Hird, Samuel Livingstone},
pdfkeywords={Markov chain Monte Carlo, Preconditioning, Bayesian inference, Bayesian Computation, Condition Number},
}

\begin{document}
\maketitle

\begin{abstract}
	We study linear preconditioning in Markov chain Monte Carlo. We consider the class of well-conditioned distributions, for which several mixing time bounds depend on the condition number  $\kappa$. First we show that well-conditioned distributions exist for which $\kappa$ can be arbitrarily large and yet no linear preconditioner can reduce it. We then impose two sets of extra assumptions under which a linear preconditioner can significantly reduce $\kappa$.  For the random walk Metropolis we further provide upper and lower bounds on the spectral gap with tight $1/\kappa$ dependence. This allows us to give conditions under which linear preconditioning can provably increase the gap. We then study popular preconditioners such as the covariance, its diagonal approximation, the hessian at the mode, and the QR decomposition. We show conditions under which each of these reduce $\kappa$ to near its minimum. We also show that the diagonal approach can in fact \textit{increase} the condition number. This is of interest as diagonal preconditioning is the default choice in well-known software packages. We conclude with a numerical study comparing preconditioners in different models, and showing how proper preconditioning can greatly reduce compute time in Hamiltonian Monte Carlo.
\end{abstract}


\tableofcontents

\section{Introduction}

Markov chain Monte Carlo (MCMC) is ubiquitous in the field of Bayesian computation. Decades of use has proven its effectiveness and longevity in the face of major changes to the structure of statistical models.
To keep pace with the increasingly complex problems faced by practitioners, the MCMC community has developed powerful and flexible algorithms within the framework of \cite{metropolis:53, hastings:70}.  Prominent examples include the Metropolis adjusted Langevin algorithm (MALA) and Hamiltonian Monte Carlo (HMC) \citep{roberts:96a, neal:11}.

Preconditioning is a technique that is now used throughout applied mathematics, including MCMC sampling, but was originally conceived to improve the stability of iterative solvers in linear algebra. The \emph{conditioning} of a particular problem corresponds to how easily it can be solved, and is quantified by a \emph{condition number}. The term `condition number' was first coined by \citet{turing:48}, although similar quantities are used in earlier works \citep{wittmeyer:36, vonneumann:47}. In its most basic form preconditioning can be thought of as transforming to a space in which the problem at hand is both easier to solve and more amenable to implementation on a computer. Assuming that the transformation is invertible, the idea is to solve the problem in the transformed space and then transform the solution back to the original space using the inverse.

Recent works by \citet{dalalyan:17, durmus2017nonasymptotic, chen:20, mangoubi2021mixing, chewi:21, lee:21, andrieu:22} and others establish non-asymptotic mixing time bounds for various MCMC samplers. The choice of distance metric $\mathcal{D}$ can vary, but excluding polylogarithmic terms each bound can be written  in the form
\[
\inf \{n>0: \mathcal{D}(P^n, \Pi) < \epsilon\} \leq C_\epsilon \kappa^\alpha d^\beta,
\]
where $P^n$ denotes the marginal distribution of the $n$th state, $\Pi$ is the equilibrium distribution, $d$ is the dimension of the state space, $\alpha$ and $\beta$ are positive constants and $C_\epsilon>0$ is an explicit constant depending on $\epsilon >0$ (see Table \ref{fig:bounds_table} for more detail).  The quantity $\kappa$, defined in equation \eqref{condition_number}, is a form of condition number. \cite{lee:21} show that for MALA and HMC there are strongly log-concave distributions for which $\alpha \geq 1$, and the best known bounds to our knowledge for these methods set $\alpha \geq 1$.  It is clear, therefore, that a smaller value for $\kappa$ results in a smaller mixing time upper bound.

In this paper we study how preconditioning an MCMC sampler affects the particular form of condition number $\kappa$ used by the MCMC community.  We focus on \emph{linear preconditioning}, meaning the chosen transformation is linear. We note that linear preconditioning is widely used in practice, but largely under-explored theoretically in MCMC (one notable exception is \cite{roberts:01}, who present asymptotic results in particular cases).  We consider common choices of linear preconditioner and explore the extent to which they will reduce the condition number of a well-conditioned probability distribution (formally defined in Section \ref{subsec:well-conditioned}).


\subsection{Notation}
\label{subsec:notation}

We use $\mathcal{X}$ to denote the Borel $\sigma$-field of $\mathbb{R}^d$, and denote the set of non-negative real numbers $\mathbb{R}^+$. Let $\Pi$ be a probability measure on $(\mathbb{R}^d,\mathcal{X})$ with Lebesgue density $\pi(x) \propto \exp(-U(x))$, where $U:\mathbb{R}^d \to \mathbb{R}$ is called the \emph{potential}. We will often call $\Pi$ \emph{the target}. The normal distribution with mean $\mu\in\mathbb{R}^d$ and covariance $\Sigma\in\mathbb{R}^{d \times d}$ will be written $\mathcal{N}(\mu, \Sigma)$ and its density at $x \in\mathbb{R}^d$ will be denoted $\mathcal{N}(x;\mu,\Sigma)$. Given a diffeomorphism $g:\mathbb{R}^d \to \mathbb{R}^d$ we define the pushforward of $\Pi$ under $g$ as the measure $g_{\#}\Pi$ such that $g_{\#}\Pi(A) := \Pi(g^{-1}(A))$ for all $A \in\mathcal{X}$. It follows that $g_{\#}\Pi$ has Lebesgue density $g_{\#}\pi(y):= \pi(g^{-1}(y))\left|\det J_{g^{-1}}(y)\right|$ where $J_{g^{-1}}(y)$ is the Jacobian of $g^{-1}$ at $y\in\mathbb{R}^d$.

We denote by $L^2(\Pi)$ the Hilbert space of functions that are square integrable with respect to $\Pi$ with inner product
\[
\langle f,g\rangle := \int f(x)g(x)\pi(x)dx,
\]
and write $L_0^2(\Pi) := \{ f \in L^2(\Pi) : \int f(x) \pi(x)dx = 0\}$. When $P$ is a $\Pi$-invariant Markov kernel we define an associated bounded linear operator $P:L^2(\Pi) \to L^2(\Pi)$ as $Pf(x) := \int P(x, dy)f(y)$. The action of $P$ on a probability measure $\nu$ is defined as $\nu P(A) := \int\int_{y\in A}\nu(dx)P(x, dy)$ for all $A\in\mathcal{X}$. We define the actions of $P^n$ for $n>1$ recursively using $P^nf(x) := P^{n - 1}(Pf(x))$ for all $x \in \mathbb{R}^d$ and $\nu P^n(A) := \int \int_{y\in A} \nu P^{n-1}(dx)P(x, dy)$ for all $A\in\mathcal{X}$.

We denote by $I$ the identity map on $L^2(\Pi)$.  For a given operator $P$ we define the \emph{Dirichlet form} $\mathcal{E}(P, f):= \langle(I - P)f,f \rangle$ for all $f \in L^2(\Pi)$. When $P$ is $\Pi$-reversible we define its \emph{right spectral gap} (from this point called \emph{spectral gap}) as
\begin{equation}\label{eqn:spectral_gap}
    \gamma := \inf_{f \in L_0^2(\Pi), f \not\equiv 0}\frac{\mathcal{E}(P,f)}{\text{Var}_\pi(f)}.
\end{equation}
The \emph{relaxation time} is simply the inverse of the spectral gap $\gamma ^ {-1}$.

Given a distance metric $\mathcal{D}$ on the space of probability measures and an initial measure $\nu_0$ the \emph{$\epsilon$-mixing time} of $P$ starting from $\nu_0$ is define for all $\epsilon > 0$ as
\begin{equation}\label{eqn:epsilon_mixing_time}
    t(\epsilon, \nu_0) := \inf \{n \in \mathbb{N} : \mathcal{D}(\nu_0P^n,\Pi)\leq \epsilon\}.
\end{equation}
We say the Markov chain has a \emph{$\beta$-warm start} if there exists a constant $\beta\in\mathbb{R}$ such that $\sup_{A\in\mathcal{X}} \nu_0(A)/\Pi(A)\leq \beta$.

The set $\{1,2,...,k\}$ is denoted as $[k]$ for any positive integer $k$. For a function $g \in C^2(\mathbb{R}^d)$ we define $\nabla g(y)\in\mathbb{R}^d$ and $\nabla^2g(y)\in\mathbb{R}^{d \times d}$ elementwise as
\[
\left(\nabla g(y)\right)_i:=\frac{\partial}{\partial y_i}g(y)\text{ }\left(\nabla^2g(y)\right)_{ij}:=\frac{\partial^2}{\partial y_i\partial y_j}g(y)
\]
for $i,j\in [d]$.

We denote by $\|.\|$ the $L_2$-norm on both $\mathbb{R}^d$ and $\mathbb{R}^{d \times d}$. For a given symmetric matrix $A$ we let $\lambda_i(A)$ be its $i$th largest eigenvalue. We define its spectral condition number as
\[
\kappa(A) := \frac{\max_{i \in [d]}|\lambda_i(A)|}{\min_{i \in [d]}|\lambda_i(A)|}.
\]
We overload the $\text{diag}$ function as follows: $\text{diag}(A)$ is the diagonal matrix that shares its diagonal with $A\in \mathbb{R}^{d \times d}$ and $\text{diag}\{f(i):i\in[d]\}\in\mathbb{R}^{d\times d}$ is the diagonal matrix whose $(i,i)$th element is $f(i)$. We denote by $GL_d(\mathbb{R})$ the set of invertible $d \times d$ matrices over $\mathbb{R}$. We use the `$\preceq$' and `$\succeq$' relations to take symmetric matrices in the following way: $A \preceq B$ (resp. $A \succeq B$) if and only if $B - A$ (resp. $A - B$) is positive semidefinite. The ordering generated by these relations is known as the Loewner order. The relations `$\prec$' and `$\succ$' are defined similarly, replacing the semidefiniteness condition with definiteness. We write $\mathbf{I}_d$ to denote the $d$-dimensional identity matrix.

For two real-valued functions $f(n)$ and $g(n)$ we say $f(n)=O(g(n))$ if there exists a universal constant $K > 0$ such that $f(n)\leq Kg(n)$. We say $f(n)=\Omega(g(n))$ if there exists a universal constant $K > 0$ such that $f(n) \geq Kg(n)$. We say $f(n) = \tilde{O}(g(n))$ if $f(n)\log^{-q}(n) = O(g(n))$ for some natural number $q$, and define $\tilde{\Omega}$ analogously.

\subsection{Main contributions}

The main contribution of our work is to illustrate the potential benefits and pitfalls of linear preconditioning of sampling algorithms. We focus attention on the class of well-conditioned target distributions, which briefly stated are probability measures $\Pi$ defined on $\mathbb{R}^d$ with density $\pi(x) \propto \exp(-U(x))$, for which the potential $U$ is twice continuously differentiable and for all $x \in \mathbb{R}^d$ satisfies
$$
m\mathbf{I}_d \preceq \nabla^2 U(x) \preceq M\mathbf{I}_d,
$$
for some finite constants $M \geq m > 0$. A more detailed description is given in Section \ref{subsec:well-conditioned}.  The associated condition number is defined as
$$
\kappa := \frac{M}{m},
$$
and it follows that $\kappa \geq 1$ with equality holding only when $\Pi$ is an isotropic Gaussian measure on $\mathbb{R}^d$.  If $X \sim \Pi$ has condition number $\kappa$, then linear preconditioning consists of applying a transformation $L \in GL_d(\mathbb{R})$ such that $Y := LX$ follows the distribution $\tilde{\Pi}$ with associated condition number $\tilde{\kappa}$.  The hope is that $\kappa \gg \tilde{\kappa} \approx 1$, as many mixing time upper bounds for sampling algorithms depend at least linearly on the condition number of the target distribution.

We informally state the main results of the paper here. Our first question concerns how broadly linear preconditioning can be effectively applied to distributions within the well-conditioned class. The below shows that further assumptions on the target distribution must be imposed to guarantee that the technique will provide benefits.

\begin{result}
    For any fixed choice of $\kappa \geq 1$, there are well-conditioned probability distributions with condition number $\kappa$ for which no choice of linear preconditioner will result in $\tilde{\kappa} \leq \kappa$.
\end{result}

\noindent This is formalised as Proposition \ref{prop:hard_target} of Section \ref{subsec:unpreconditionable}.  The result shows that standard assumptions made in log-concave sampling are not sufficient to guarantee that linear preconditioning can reduce the condition number of a given sampling problem.  In the next two results, however, we provide additional assumptions under which guarantees can be provided.

\begin{result}
    If $\Pi$ is such that $\nabla^2 U$ is of \emph{additive} form, meaning it can be written $\nabla^2 U(x) = A + B(x)$ in such a way that $\sup_x\|B(x)\| < \epsilon$ for some $\epsilon\geq 0$, then there are choices of linear preconditioner for which
    $$
    \tilde{\kappa} \leq 1 + f(\epsilon),
    $$
    where $f$ is explicit, monotonically increasing and $\lim_{\epsilon\to 0 }f(\epsilon) = 0$. If further conditions are imposed on $B(x)$ then tighter bounds can be applied to $\tilde{\kappa}$.  Valid choices of preconditioner are:
    \begin{enumerate}
        \item[(i)] $L = A^{1/2}$ if known
        \item[(ii)] The Fisher matrix $L = \mathbb{E}_\pi[\nabla^2 U(x)\nabla^2 U(x)^T]^{1/2}$
        \item[(iii)] The choice $L = \textup{Cov}_\pi[X]^{-1/2}$ (under additional technical conditions)
    \end{enumerate}
\end{result}

The first part of the result is formalized in Section \ref{subsec:additive_hessian} in Theorems \ref{thm:kappa_L_bound_1}-\ref{thm:kappa_L_bound_3}, and the second part is also contained in the same section.  In Theorem \ref{thm:kappa_L_bound_1} precise control on the eigenvalues and eigenvectors are required of $\nabla^2U(x)$, after which a result can be proved using elementary linear algebra.  In Theorems \ref{thm:kappa_L_bound_2} \& \ref{thm:kappa_L_bound_3}, we relax this requirement, controlling only the operator norm $\|B(x)\|$ together with an `eigengap' condition in the case of Theorem \ref{thm:kappa_L_bound_2}. This weakening relies on powerful results from matrix perturbation theory, such as a recent variant of the Davis--Kahan theorem proven in \cite{yu:15}. Our bounds are tight in the sense that there are examples for which they can be saturated.  The second part of the result is formalised in Sections \ref{subsubsec:fisher_matrix} and \ref{subsubsec:target_covariance}. It shows that popular strategies of preconditioner, which are often approximated using pilot runs or early samples from the Markov chain during the warm up/burn in phase, are effective when the Hessian is of additive form. We also give empirical examples of performance improvements for models with this structure in Section \ref{subsec:experiment_additive}.

When the additive condition of Result 2 does not hold we offer an alternative structure that is also common in probabilistic models, leading to the next result.

\begin{result}
    If $\Pi$ is such that $\nabla^2 U$ is of \emph{multiplicative} form, meaning it can be written $\nabla^2 U(x) = Z^T \Lambda(x) Z$ with $Z \in \mathbb{R}^{n \times d}$ and $\Lambda(x) \in \mathbb{R}^{n \times n}$, then
    \begin{enumerate}
        \item[(i)] 
        Choosing either $L = (Z^TZ)^{1/2}$ or $L = R$ where $QR$ is a reduced QR decomposition of $Z$ gives
        \[
        \tilde{\kappa} \leq\frac{\sup_{x\in\mathbb{R}^d}\lambda_1(\Lambda(x))}{\inf_{x\in\mathbb{R}^d}\lambda_d(\Lambda(x))} 
        \]
        \item[(ii)] Choosing $L = (Z^T \Lambda(x^*) Z)^{1/2}$ gives
        \[
        \tilde{\kappa} \leq \frac{\sup_{x\in\mathbb{R}^d}\lambda_1(\Lambda(x^*)^{-\frac{1}{2}}\Lambda(x)\Lambda(x^*)^{-\frac{1}{2}})}{\inf_{x\in\mathbb{R}^d}\lambda_d(\Lambda(x^*)^{-\frac{1}{2}}\Lambda(x)\Lambda(x^*)^{-\frac{1}{2}})}  
        \]
        for any $x^* \in \mathbb{R}^d$.
    \end{enumerate}
\end{result}

This result is formalised in Propositions \ref{prop:mult_dalalyan} and \ref{prop:mult_Hessian_at_point} in Section \ref{subsec:multiplicative_hessian}.  The proofs in this section rely on a generalisation of Ostrowski's theorem for congruence transformations of symmetric matrices to the rectangular case, which was proven in \cite{higham:98a} (see also \cite{ostrowski1959quantitative}).  In theory $n$ can be any natural number larger than $d$, but in practice it will often be the number of data points and $Z$ a design matrix, in which case the multiplicative structure will be familiar to practitioners from generalized linear models and related settings.  The QR decomposition is a popular strategy in applications and is discussed in Section \ref{subsubsec:QR}, while if $x^*$ is chosen to be the mode of $\pi$ then the second choice corresponds to the Hessian evaluated at the mode. More details are given in Section \ref{subsec:multiplicative_hessian}.  These strategies are known to be effective in practice, but in Section \ref{subsec:experiment_multiplicative} we show empirically that for binomial regression there are settings for which the Hessian at the mode should be preferred.

The aforementioned results show that appropriate linear preconditioning can reduce the condition number, which in turn reduces mixing time upper bounds. In the next result we go further, showing rigorously that under appropriate conditions the spectral gap of the random walk Metropolis can be explicitly improved through linear preconditioning.

\begin{result}
    If $\Pi$ is such that $\nabla^2 U$ is of \emph{additive} form, meaning it can be written $\nabla^2 U(x) = A + B(x)$ in such a way that $\sup_x\|B(x)\| < \epsilon$ for some $\epsilon\geq 0$, then the spectral gap $\gamma_\kappa$ of the random walk Metropolis satisfies
    \[
    C\xi\exp(-2\xi)\cdot\frac{1}{\kappa}\cdot\frac{1}{d}\leq \gamma_\kappa \leq (1 + 2\epsilon)\frac{\xi}{2}\cdot\frac{1}{\kappa}\cdot\frac{1}{d}
    \]
    where $C = 1.972 \times {10}^{-4}$, and $\xi > 0$ is a constant depending on the random walk step-size.  If further conditions are imposed on $B(x)$ then tighter bounds can be applied to $\kappa$.
\end{result}

The result is formalised in Section \ref{subsec:spectral_gap}.  The lower bound is taken directly from \cite{andrieu:22}, but the upper bound is new and relies on the tighter control on the spectrum of $\nabla^2U$ imposed by the additive structure.  Since the condition number dependence is the same in both the upper and lower bounds, we can as a direct consequence of Results 2 \& 4 show that if $\kappa$ is sufficiently large and $\nabla^2 U$ has additive structure then the spectral gap is guaranteed to increase under appropriate linear preconditioning.  See Section \ref{subsec:spectral_gap} and in particular Corollary \ref{cor:improved_gap} for more detail.




Our final theoretical result provides a word of caution for the popular strategy of \emph{diagonal} preconditioning, in which $L = \text{diag}(\text{Cov}_\pi[X])^{-1/2}$, or some approximation of this quantity estimated from samples.  The approach is popular because $L$ is now diagonal, meaning that operations involving it will be $O(d)$ instead of $O(d^2)$. Diagonal preconditioning is the default choice in some probabilistic programming software (e.g. \cite{carpenter2017stan}). In Hamiltonian Monte Carlo, for example, this strategy corresponds to the dominant choice of diagonal mass matrix. It is generally understood that if $\Pi$ exhibits large correlations then the diagonal strategy may not reduce the condition number by much if at all, but in fact we show below that it can explicitly be worse than doing nothing at all.

\begin{result}
    There are well-conditioned distributions for which applying the preconditioner $L = \textup{diag}(\textup{Cov}_\pi[X])^{1/2}$ gives
    $$
    \tilde{\kappa} > \kappa,
    $$
    while choosing $L = \textup{Cov}_\pi[X]^{-1/2}$ results in $\tilde{\kappa} \approx 1$.
\end{result}

The above highlights that when large correlations are present then the practitioner is better off applying no preconditioning than using the diagonal strategy.  We show that this strategy can result in demonstrably worse practical sampling performance in Section \ref{subsec:experiment_counterproductive} using a simple 5-dimensional Gaussian example.

A final empirical contribution of the paper is to demonstrate how linear preconditioning can reduce computational costs in Hamiltonian Monte Carlo. In Section \ref{subsec:experiment_hmc} we provide two examples showing that as the condition number increases then roughly $\kappa^{1/2}$ leapfrog steps must be taken in order to achieve a comparable effective sample size for the algorithm, and that effective linear preconditioning reduces the computational cost per step of the algorithm by this order of magnitude.  This experiment was suggested by a reviewer of an earlier version of the manuscript.

\section{Preconditioning in Markov chain Monte Carlo}

Suppose $\Pi$ is a measure on $(\mathbb{R}^d,\mathcal{X})$ with density $\pi(x) \propto e^{-U(x)}$ for some \emph{potential} $U:\mathbb{R}^d \to \mathbb{R}$. The aim of MCMC is to sample from $\Pi$ so as to infer its properties. In particular we are interested in estimating expectations $\mathbb{E}_\pi[f(X)]$ of some test function $f$. Given samples $\{X_i\}_{i = 1} ^ N$ we achieve this by forming the estimator $\hat{f}_N := N^{-1}\sum_{i = 1} ^ N f(X_i)$. In many situations of interest drawing independent samples from $\Pi$ is computationally infeasible. The lifeblood of MCMC is in creating and analyzing algorithms which produce Markov chains whose equilibrium distribution is $\Pi$. We can then use the samples from these Markov chains in estimators such as $\hat{f}_N$.

\subsection{Markov chain Monte Carlo}

Constructing a Markov chain with a chosen equilibrium distribution can be achieved using the \emph{Metropolis--Hastings} filter \citep{metropolis:53,hastings:70}. Given a Markov kernel $Q_\theta$ and its density $q_\theta(x \to .)$ we can construct a $\Pi$-invariant Markov chain using the pseudocode in Algorithm \ref{alg:one}. To apply step 2 to an existing Markov kernel $Q_\theta$ is to \emph{Metropolize} it. 

\begin{algorithm}
\caption{Metropolized Markov chain}\label{alg:one}
\SetKwInOut{Input}{input}\SetKwInOut{Output}{output}
\Input{Chain length $N$, Initial distribution $\mu$, Initial state $X_0 \sim \mu$, Proposal parameters $\theta$}
\Output{$\Pi$-invariant Markov chain $\{X_i\}_{i = 1} ^ N$}
\For{$i \in \{0,...,N - 1\}$}{
  \begin{enumerate}
      \item \textbf{Proposal}: Propose a new state $X'_i \sim Q_\theta(X_i \to .)$.
      \item \textbf{Metropolis--Hastings filter}: Draw $U \sim U[0, 1]$ and calculate
      \[
      \alpha(X_i \to X'_i) := \min\left\{1, \frac{\pi(X'_i)q_\theta(X'_i \to X_i)}{\pi(X_i)q_\theta(X_i \to X'_i)}\right\}
      \]
      If $U \leq \alpha(X_i \to X'_i)$ set $X_{i + 1} = X'_i$, otherwise set $X_{i + 1} = X_i$.
  \end{enumerate}
}
\end{algorithm}

When $q_\theta(x \to .) := \mathcal{N}(\hspace{0.1cm}.\hspace{0.1cm}; x, \sigma ^ 2 \mathbf{I}_d)$ Algorithm \ref{alg:one} defines the \emph{random walk Metropolis}  algorithm (RWM). Here $\sigma > 0$ is a \emph{step-size} and must be tuned by the practitioner. When $q_\theta(x \to .) := \mathcal{N}(\hspace{0.1cm}.\hspace{0.1cm}; x - \sigma ^ 2\nabla U(x)/2 , \sigma ^ 2 \mathbf{I}_d)$ Algorithm \ref{alg:one} defines the \emph{Metropolis-adjusted Langevin algorithm} (MALA). Note that this proposal mechanism follows a single step with step-size $\sigma ^ 2$ of the Euler--Maruyama discretisation of the overdamped \emph{Langevin diffusion}, described through the dynamics
\begin{equation}\label{eq:langevin_diffn}
    dX_t=-\frac{1}{2}\nabla U(X_t)dt + dB_t,
\end{equation}
where $(B_t)_{t\geq 0}$ is a standard Brownian motion on $\mathbb{R}^d$ (we drop the word \emph{overdamped} from this point forward). The Langevin diffusion associated with a given potential $U$ admits $\Pi$ as equilibrium distribution under mild conditions (e.g. \cite{roberts:96a}).

The idea underpinning HMC \citep{duane:87, neal:11} is to extend the state space to $\mathbb{R} ^ d \times \mathbb{R} ^ d$ to include a momentum vector, and then viewing the states in the Markov chain as the positions and momenta in physical space of a particle moving in the potential $U$. To propose a new state $X'_i$ we simply evolve a particle with initial position $X_i$ according to a discretisation of Hamiltonian dynamics with random initial momentum and potential $U$. Specifically we propose $X'_i = x_T$ using the final value $(x_T, p_T)$ of the trajectory described by the Hamiltonian system:
\begin{equation}\label{eq:hamiltonian_dynamics}
    \frac{dx_t}{dt} = \nabla_p H(x_t, p_t)\text{, }\frac{dp_t}{dt}=-\nabla_x H(x_t, p_t)
\end{equation}
with initial conditions $x_0 = X_i$, $p_0 \sim \mathcal{N}(0, \mathbf{I}_d)$. The function $H:\mathbb{R} ^ d \times \mathbb{R} ^ d \to \mathbb{R}$ is the \emph{Hamiltonian} $H(x, p) := U(x) + K(p)$ where $K(p) = p^Tp/2$ is the \emph{kinetic energy}. Note that if we marginalise the density proportional to $\exp(-H(x, p))$ over $p$ we are left with $\pi$.  See \citet{neal:11,betancourt:17} for a more thorough exposition, \citet{livingstone:19} for extensions to different forms of $K(p)$, and \citet{chen:14} for a stochastic gradient version of the algorithm.  Convergence properties are discussed in \cite{livingstone:19geometric,durmus:20}.  For a more general discussion of schemes involving augmentation of the state vector see \citet{andrieu:20,glatt:20}.

\subsection{Well-conditioned distributions and the condition number}
\label{subsec:well-conditioned}

\subsubsection{Background assumptions on $\Pi$}

Here we introduce a background assumption which is assumed for most theorems in the text.
\begin{ass}\label{ass:m sc and M smoothness} The potential $U$ is in $ C^2$, and there exist constants $M \geq m > 0$ such that 
\begin{equation}
    m\mathbf{I}_d\preceq\nabla^2U(x)\preceq M\mathbf{I}_d
\end{equation}
for all $x\in \mathbb{R}^d$.
\end{ass}

The lower bound on the Hessian $\nabla^2 U$ in Assumption \ref{ass:m sc and M smoothness} is known as \emph{$m$-strong convexity} of $U$ and the upper bound is commonly known as \emph{$M$-smoothness} of $U$. Note that $M$-smoothness holds if and only if $\nabla U$ is $M$-Lipschitz. 

\subsubsection{The condition number}

When $U$ satisfies Assumption \ref{ass:m sc and M smoothness} the \textit{condition number} is defined as
\begin{equation}\label{condition_number}
    \kappa := \frac{M}{m}.
\end{equation}
When $\Pi$ is Gaussian with covariance $\Sigma_\pi\in\mathbb{R}^{d\times d}$,
then $\kappa = \lambda_1(\Sigma_\pi)/\lambda_d(\Sigma_\pi)$, which is the spectral condition number  of $\Sigma_{\pi}$. When $U\in C^2$ the condition number can be equivalently defined as
\begin{equation}\label{condition_number_sup}
    \kappa := \sup_{x\in\mathbb{R}^d}\|\nabla^2 U(x)\|\sup_{x\in\mathbb{R}^d}\|\nabla^2 U(x)^{-1}\|.
\end{equation}
It is bounded below by 1, and a larger condition number traditionally signifies a harder problem that is less amenable to implementation on a computer. 

The significance of $\kappa$ in the context of sampling is demonstrated by its presence in bounds on quantities which govern the performance of MCMC algorithms, such as the relaxation time and the $\epsilon$-mixing time. The relaxation time is defined below equation (\ref{eqn:spectral_gap}) and the $\epsilon$-mixing time is defined in equation (\ref{eqn:epsilon_mixing_time}). The two quantities provide some measure of the length of time to run the MCMC algorithm. Note that the relaxation time is not defined with respect to the initialisation of the Markov chain, whereas the $\epsilon$-mixing time is. The two quantities are related since, for instance, a small relaxation time will imply a small $\epsilon$-mixing time if the chain is positive (i.e. the spectrum of $P$ is contained in $[0,1]$) and is well initialised, using e.g. a $\beta$-warm start as defined in Section \ref{subsec:notation}. We present a selection of bounds on these quantities in Table \ref{fig:bounds_table}. Each is polynomial in both the dimension \emph{and} the condition number. The selection is not exhaustive, we merely present it to highlight the ubiquity of the condition number.
\begin{table}
    \centering
    \renewcommand{\arraystretch}{2.2}
    \begin{tabular}{|c|c|c|}
\cline{2-3} \cline{3-3} 
\multicolumn{1}{c|}{} & Relaxation Time & $\epsilon$-Mixing Time\tabularnewline
\hline 
 &  & $\tilde{O}\left(\kappa d^{\frac{1}{2}}\right)^{\bigtriangledown}$ \citet{wu:22}\tabularnewline
\cline{3-3} 
Upper Bound & $O(\kappa d)^{\bigtriangleup}$ \citet{andrieu:22} & $\tilde{O}\left(\kappa d\log\frac{1}{\epsilon}\right)^{\bigtriangleup}$ \citet{andrieu:22}\tabularnewline
\cline{3-3} 
 &  & $\tilde{O}\left(\kappa d^{\frac{2}{3}}\log\frac{1}{\epsilon}\right)^{\bigcirc}$ \citet{chen:20}\tabularnewline
\hline 
Lower Bound & $\Omega\left(\frac{\kappa d}{\log d}\right)^{\bigtriangledown}$ \citet{lee:21}& $\Omega\left(\frac{\kappa d}{\log^{2}d}\right)^{\bigtriangledown}$ \citet{lee:21}\tabularnewline
\hline 
\end{tabular}
    \caption{Recently published bounds on the relaxation time and $\epsilon$-mixing time of various MCMC samplers. The $\tilde{O}$ denotes that the bound excludes poly-logarithmic terms. Superscript $\bigtriangleup$ denotes a bound for RWM, superscript $\bigtriangledown$ for MALA, and superscript $\bigcirc$ for HMC. The bound in \citet{chen:20} holds under the additional assumptions that $\nabla^2U$ is $L_H$-Lipschitz with $L_H^{2/3} = O(M)$, that $\kappa = O(d^{2/3})$, and that the chain is started from a warm start with constant $W = O(\exp(d^{2/3}))$.}
    \label{fig:bounds_table}
\end{table}
There also exist numerous bounds on the relaxation and $\epsilon$-mixing time of the unmetropolized counterparts of the algorithms in Table \ref{fig:bounds_table} (see e.g.  \citet{mangoubi:18,dalalyan:17}.

\subsection{Preconditioning}

Assuming that $X\sim\Pi$ we let $\tilde{E}:=(\mathbb{R}^d,\mathcal{X})$
be the space generated by the diffeomorphism $g:\mathbb{R}^d\to\mathbb{R}^d$
such that $Y:=g(X)$ is distributed according to the pushforward $\tilde{\Pi}:=g_{\#}\Pi$ which is a distribution on $\tilde{E}$ with density $\tilde{\pi}$.
Preconditioning is the act of identifying a transformation $g$ with the hope that
that $\tilde{\Pi}$ admits a faster mixing MCMC algorithm than $\Pi$. Linear preconditioning is the case in which $g$ is linear and therefore encoded in a matrix. We would like the condition number of $\tilde{\Pi}$ to be significantly lower than that of $\Pi$. We would also like that for two samples $Y$ and $X$
sampled approximately from $\tilde{\Pi}$ and $\Pi$, the increase
in sample quality from $X$ to $g^{-1}(Y)$ is not outweighed by the computational complexity of sampling $Y$ and calculating $g^{-1}(Y)$.

\subsection{Preconditioning in optimisation} 
\label{subsec:optimisation}

The method of preconditioning to reduce a condition number and improve algorithmic efficiency also exists in the field of optimisation, i.e. finding $x^* := \arg\min_x f(x)$ subject to a problem-specific set of constraints.  In this subsection we review relevant literature and make connections to preconditioning in sampling.

One example problem is to seek the minimiser $x^* \in \mathbb{R}^n$ of the function $f(x):= (1/2)x^T A x - b^T x$ for a positive definite matrix $A\in\mathbb{R}^n$ and a vector $b\in\mathbb{R}^n$. Setting $\nabla f(x)$ to 0 we can immediately see that $x^* = A^{-1}b$ and hence the problem is equivalent to solving the matrix inversion problem $Ax^* = b$. This problem has acquired the condition number $\kappa(A) = \|A\|\|A^{-1}\|$. Many iterative methods have been proposed for this problem, whose performance typically depends on $\kappa(A)$. One such method is the Conjugate Gradient method, which has the property of converging exactly in $\leq n$ steps \citep[Chapter 5]{nocedal:06}. In fact we can use the bound
\begin{equation*}
    \|x_{k+1} - x^*\|_A \leq 2\left(\frac{\sqrt{\kappa(A)} - 1}{\sqrt{\kappa(A)} + 1}\right)^k\|x_0 - x^*\|_A,
\end{equation*}
where $\|x\|_A:=\sqrt{x^TAx}$, to control the optimiser error of the iterates. Preconditioning methods for conjugate gradient have therefore been developed to reduce the condition number and accelerate convergence (see e.g. \cite{eisenstat1981efficient}). 

In the more general setting of minimising an $M$-smooth, $m$-strongly convex objective function $f$, the spectrum of the Hessian of $f$ is entirely contained within $[m, M]$ at every point in its domain. This spectral information is used to define a new condition number $\kappa := M / m$. Note that this is exactly the condition number defined for the problem of sampling from a density with potential $f$, see equation (\ref{condition_number}).  Popular general purpose iterative methods are used to execute this task.  These methods often use gradient evaluations of $f$ in both their operation and in stopping conditions such as $\|\nabla f (x)\|\leq \eta$ for some user-chosen $\eta \in [0, \infty)$. The constants $m$ and $M$ are informative, as for instance
\begin{equation*}
    \frac{1}{2}\frac{1}{M}\|\nabla f (x)\|^2 \leq f(x) - f(x^*) \leq \frac{1}{2}\frac{1}{m}\|\nabla f (x)\|^2
\end{equation*}
for all $x$ in the domain of $f$ \citep[Section 9.1.2]{boyd:04}. When $\kappa$ is closer to 1, therefore, the stopping condition $\|\nabla f (x)\|\leq \eta$ gives tighter control on distance to the optimum. The condition number can also be used to quantify the efficiency of the line-search methods used within these approaches, since it provides information on the regularity of the sub-level sets of $f$ \citep[Section 9.1.2]{boyd:04}. 

The time-complexity of general purpose, iterative optimisation routines is typically strongly dependent on the condition number. The complexity of gradient descent, for example, is $O(\kappa)$ (suppressing dependencies on initialisation and error) \citep[Section 9.3.1]{boyd:04}. \cite{bubeck:15} provides the convergence bound
\begin{equation} \label{eq:gradient_descent}
\|x_t - x^*\|^2 \leq \exp \left( -\frac{4t}{\kappa+1} \right)\|x_1 - x^*\|^2,
\end{equation}
where $x_t$ denotes the $t$th iterate of the algorithm \citep[Section 3.4]{bubeck:15}.  Nesterov's accelerated gradient descent, by contrast, has a time-complexity of $O(\sqrt{\kappa})$, which is known to be optimal among first order methods following the work of \cite{nemirovski:83} (see e.g. Theorems 3.15 \& 3.18 of \cite{bubeck:15}). 

Many linear preconditioning strategies have been proposed for optimisation. In the context of solving the linear system \cite{chen2005matrix} provides a book-level treatment. In that setting  the condition number can be reduced to 1 using matrix decompositions such as LU and QR approaches, but these are often too expensive.  Various approximate strategies to compute either $A^{-1/2}$ or $A^{-1}$ have therefore been devised, such as banded matrices \citep[Chapter 5] {chen2005matrix}, the incomplete LU and Cholesky factors \citep{meijerink1977iterative,jones1995improved}, the schur complement \citep{zhang2006schur}, sparse approximate inverses \citep{grote1997parallel}, hierarchical off-diagonal low rank matrices and hierarchical semi-separable matrices \citep{massei2020hm}, among others.  A common approach is the Jacobi preconditioner $\text{diag}(A)^{-1/2}$, whose effects on the conditioner number are studied in \cite{van1969condition} and extended by \cite{demmel2023nearly} to block diagonal preconditioners. The Jacobi approach generally works well if $A$ is diagonally dominant.  \cite{mandel1990block} gives conditions under which the Schur complement will more significantly reduce the condition number than the block-diagonal approach.  Typically some case-specific knowledge of the system is exploited to choose design parameters in these approaches, such as the number of bands/blocks to choose, the level of sparsity or the rank of off-diagonal sub-matrices in question.

In the more general setting many linear preconditioners have also been proposed (see e.g. \citet{giselsson:14, pock:11, wen:17}).  Effects on a condition number, however, are not always studied.  One exception is \cite{amelunxen:20}, who show that the lower bound on a variant of the condition number called `Renegar's condition number' changes under preconditioning with a random orthogonal matrix.  An alternative strategy in the general case is to use second order approaches such as Newton and Quasi-Newton methods, which use or approximate the Hessian of $f$ at the most recent iteration to inform the descent direction.  Samplers which use point-wise second order information also exist, see e.g. \cite{girolami:11}.  In the sampling context, however, the requirement to satisfy $\Pi$-invariance at least approximately often results in either the need for third order information or other additional expenses, which can incur significant computational cost.  In this work we focus on methods which do not require point-wise evaluation of second or higher order information at each iteration, but note this to be an interesting topic for future work (see Section \ref{subsec:nonlinear}).

This paper focuses on the problem of sampling, but we note that the results that only concern the condition number can be readily applied to convex optimisation with little translation, as the condition number is defined similarly in both cases.

\section{Linear preconditioning}\label{Linear preconditioning}

In linear preconditioning we seek a matrix $L\in GL_{d}(\mathbb{R})$
such that the preconditioner defined by $g(X) := LX$ produces a pushforward $\tilde{\Pi} = g_{\#}\Pi$ with small condition number $\tilde{\kappa}$. From the definition of the condition number in equation (\ref{condition_number_sup}) we see that
\begin{equation}\label{conditioned_number}
    \tilde{\kappa}=\sup_{y\in\mathbb{R}^{d}}\left\|\nabla^{2}\tilde{U}(y)\right\|\sup_{y\in\mathbb{R}^{d}}\left\|\nabla^{2}\tilde{U}(y)^{-1}\right\|=\sup_{x\in\mathbb{R}^{d}}\left\|L^{-T}\nabla^{2}U(x)L^{-1}\right\|\sup_{x\in\mathbb{R}^{d}}\left\|L\nabla^{2}U(x)^{-1}L^{T}\right\|
\end{equation}
 where $\tilde{U}:\mathbb{R}^{d}\to\mathbb{R}$ is the potential associated with $\tilde{\Pi}$. For a concrete example, Figure \ref{fig:preconditioning_demonstration} shows two contour plots. The red contour plot is generated by the posterior of a Bayesian logistic regression, and the green plot shows the distribution after preconditioning with an appropriate linear preconditioner. The red plot shows that most of the mass of the corresponding posterior is confined to a thin wedge, meaning that only a narrow range of tuning parameters will lead to a successful MCMC algorithm attempting to sample from this posterior. The green plot shows that the corresponding posterior mass is more evenly spread. This is one of the ways in which linear preconditioning can be seen to improve the ease with which we sample.

\begin{figure}
    \centering
    \includegraphics[scale = 0.15]{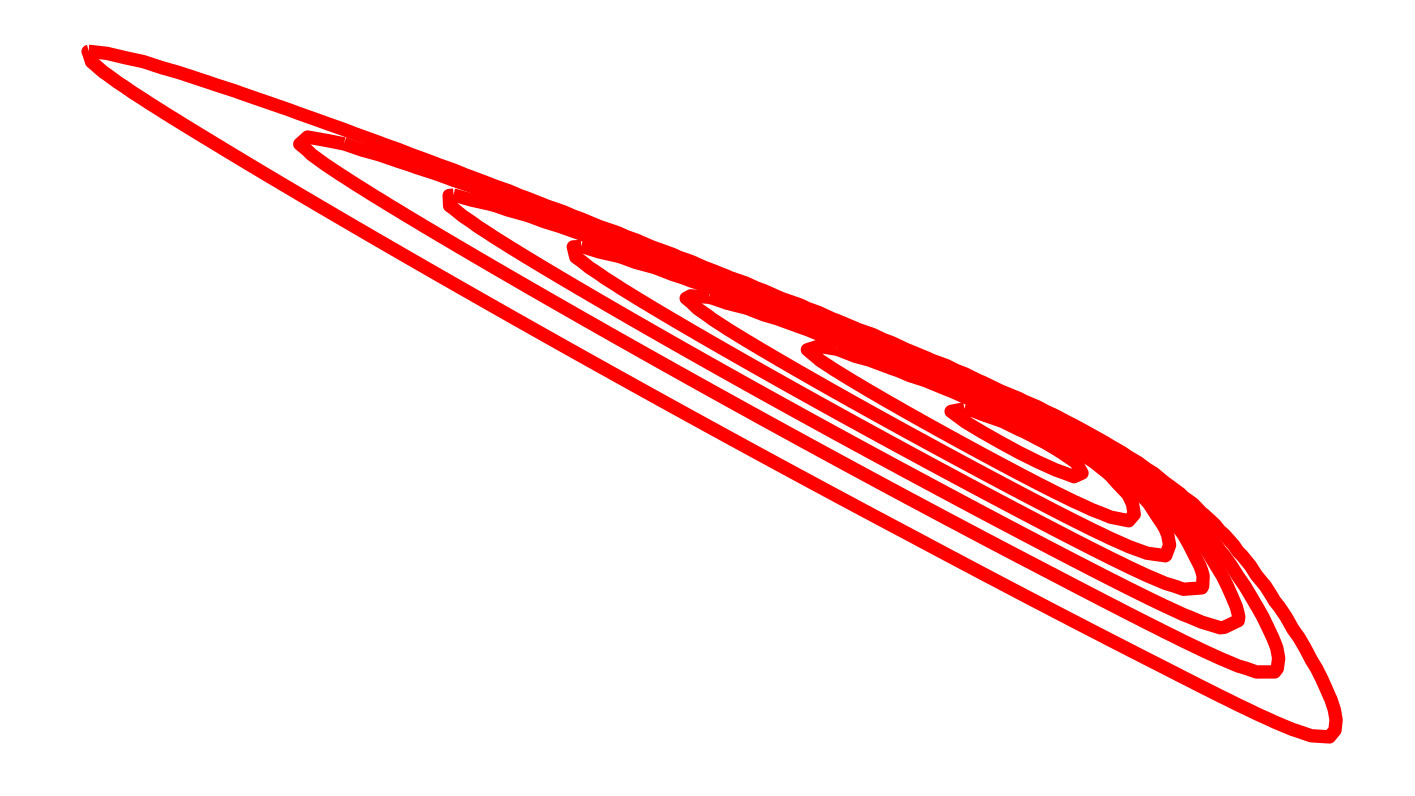}\includegraphics[scale = 0.15]{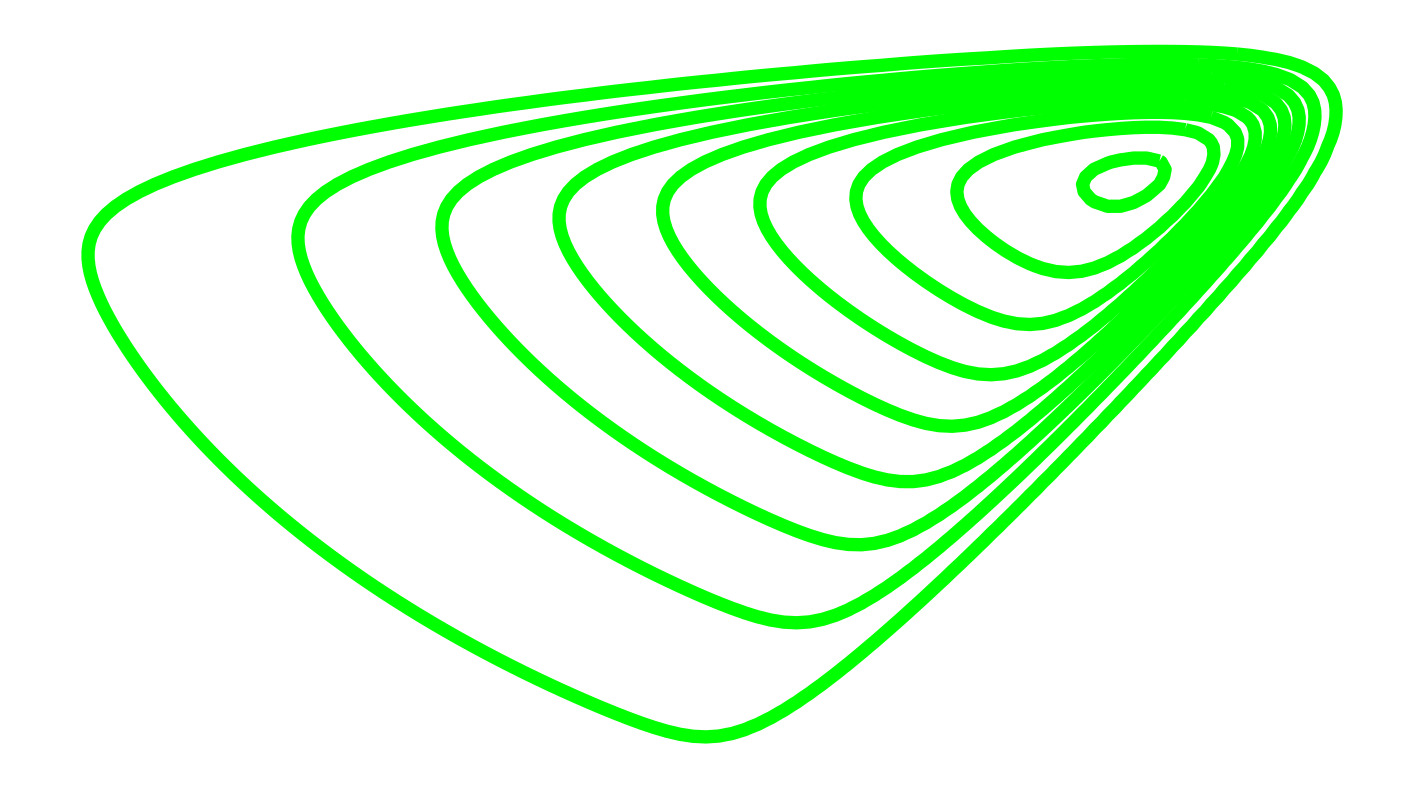}
    \caption{A red contour plot (left) showing the contours of a Bayesian logistic regression posterior, and a green contour plot (right) showing the contours of the posterior preconditioned with an appropriate matrix $L \in\mathbb{R}^{d \times d}$}
    \label{fig:preconditioning_demonstration}
\end{figure}

 Often linear preconditioning is encountered as a modification to the proposal distribution of a canonical MCMC algorithm and not as a transformation. In every case these modifications can be shown to be equivalent to making a linear transformation to $\Pi$ and then running the canonical algorithm \emph{without} preconditioning on the transformed distribution $\tilde{\Pi}$. For instance, preconditioned MALA has proposal
 \[
 Y' := Y + \sigma^2A\nabla_y\log\tilde{\pi}(Y)+\sqrt{2\sigma^2}A^\frac{1}{2}\xi
 \]
 where $\sigma^2>0$ is the step-size, $\tilde{\pi}$ is the target density, $\xi\sim\mathcal{N}(0, \mathbf{I}_d)$, and $A$ is a positive definite `preconditioning matrix'. This is equivalent to making a proposal
 \[
 X' = X +\sigma^2\nabla_x\log\pi(X)+\sqrt{2\sigma^2}\xi
 \]
 under the transformation $Y=LX$ where $L = A^{1/2}$ and $\pi(X)=\tilde{\pi}(Y)|\det A^{1/2}|$. That the acceptance probabilities are equal dictates that the chains $\{Y_i\}$ and $\{X_i\}$ are \emph{isomorphic} in the sense of \citet[Appendix A]{johnson:12} which means, amongst other things, that they have the same spectral gap. Similarly the preconditioned random walk Metropolis has proposal $Y' \sim N(Y, \sigma^2A)$ and in preconditioned HMC we choose the momentum distribution to be $K(p) = p^T A p/2$ for $M = A^{-1}$ is usually called the mass matrix.  We frame preconditioning in terms of a transformation because it suits our analysis, but the reader more familiar with preconditioning as changing the proposal distribution of the MCMC algorithm can readily do so, the conclusions we draw are naturally applicable to both ways of thinking about the problem.
 
 \subsection{Unpreconditionable Distributions} \label{subsec:unpreconditionable}
 Methods to adaptively seek linear
preconditioners are implemented within the MCMC samplers provided by the
major software packages. For instance the HMC sampler in the popular
statistical modelling platform Stan  offers the ability
to infer the target covariance $\Sigma_{\pi}\in\mathbb{R}^{d\times d}$
giving an estimator $\hat{\Sigma}_{\pi}\in\mathbb{R}^{d\times d}$
to use as the inverse mass matrix \citep{carpenter2017stan}. This is equivalent to linear preconditioning
with $L=\hat{\Sigma}_{\pi}^{-1/2}$. Since computational effort is required to infer $\hat{\Sigma}_\pi$ the idea is that preconditioning with $L=\hat{\Sigma}_{\pi}^{-1/2}$ is better than doing nothing. As we show in the below proposition, however, even within the well-conditioned class of models satisfying Assumption \ref{ass:m sc and M smoothness}, this will not always be the case, as shown in the below proposition.

\begin{prop}\label{prop:hard_target}
There exist distributions satisfying Assumption \ref{ass:m sc and M smoothness} for which any non-orthogonal linear preconditioner will cause the condition number to increase.
\end{prop}

One example for which no linear preconditioner can reduce the condition number is the distribution with potential
\begin{equation}\label{eq:hard_potential}
    U(x)=\frac{m-M}{2}\left(\cos x_1+\cos x_2\right)+\frac{M+m}{2}\left(\frac{x_1^{2}}{2}+\frac{x_2^{2}}{2}\right)
\end{equation}
for $x\in\mathbb{R}^2$. For such a target the condition number $\tilde{\kappa}$ is bounded below by $\kappa(LL^T)\kappa$, where $\kappa(LL^T)$ is the spectral condition number of $LL^T$. The Hessian of $U$ is in the form $\text{diag}\{f(x),f(y)\}$
where $f$ ranges freely in $[m,M]$. The lower bound $\kappa(LL^T)\kappa$ highlights that any preconditioning causes the target to be more ill-conditioned by an amount exactly proportional to $\kappa(LL^T)$. The full proof of Proposition \ref{prop:hard_target} can be found in Appendix \ref{proof:hard_target}.


The conclusion of Proposition \ref{prop:hard_target} is that to establish positive results about the benefits of linear preconditioning further conditions must be imposed on $U$. Specifically we seek conditions under which the behaviour demonstrated by the example above is precluded.  In the following subsections we impose two different sets of conditions under which positive results can be established, and provide examples of commonly used preconditioning strategies which can then be shown to reduce the condition number.

Throughout Sections \ref{subsec:additive_hessian} and \ref{subsec:multiplicative_hessian} $\lambda_i(x)$ denotes the $i$-th largest eigenvalue of $\nabla ^ 2 U(x)$ and $v_i(x)$ the corresponding normalised eigenvector. Since the Hessian is everywhere symmetric its eigenvectors are orthogonal for a fixed $x \in \mathbb{R}^d$. At some points it will be useful to assume that $L$ is symmetric. The following proposition shows that this can be done without loss of generality.
\begin{prop}\label{prop:symmetry_invariance}
    There exists an $\tilde{L}\in GL_{d}(\mathbb{R})$ such that the condition number after linear preconditioning with a matrix $L\in GL_{d}(\mathbb{R})$
is equal to the condition number after linear preconditioning
with $\tilde{L}$. Moreover such a $\tilde{L}$ is
symmetric, positive definite, and has eigenvalues equal to the singular
values of $L$.
\end{prop}
A proof can be found in Appendix \ref{proof:symmetry_invariance}. If $L$ is assumed to be symmetric we can then denote its $i$-th eigenvalue as $\sigma_i$, with associated eigenvector $u_i$.

\subsection{Linear preconditioning for additive Hessians}\label{subsec:additive_hessian}

We call a Hessian \emph{additive} if it has the form
\begin{equation}\label{additive hessian definition}
    \nabla^{2}U(x)=A+B(x)
\end{equation} where $A,B(x)\in\mathbb{R}^{d\times d}$
are symmetric and $\|B(x)\|$ is `small' for all $x \in\mathbb{R}^d$. Examples of models whose potentials have Hessians in
this form include: Gaussians ($B(x)\equiv0$), strongly log-concave
mixture of Gaussians \citep[Section~6.1]{dalalyan:17}, and Bayesian Huberized regressions
with strongly log-concave and smooth priors \citep{rossett:04}. The results in this section are presented in generality, but the assumptions under which they hold are particularly appropriate for models with an additive Hessian.

We first present a general result under the following assumptions on the eigenstructure of $\nabla^2 U$ and $L$.

\begin{ass}\label{ass:explicit eigenvalue} There exists an $\epsilon \geq 0$ such that
    \[
    (1 + \epsilon) ^ {-1} \leq \frac{\lambda_i(x)}{\sigma_i^2} \leq 1 + \epsilon
    \]
    for all $i \in [d]$ and $x \in \mathbb{R}^d$, where $\sigma_i > 0$ is the $i$th eigenvalue of the preconditioner $L$.
\end{ass}

\begin{ass}\label{ass:explicit eigenvector} There exists a $\delta \geq 0$ such that $v_i(x)^Tu_i \geq 1 - (1-\sqrt{1-\delta}) ^ 2$ for all $i \in [d]$ and $x \in \mathbb{R}^d$.
\end{ass}

\begin{thm}\label{thm:kappa_L_bound_1}
Let $\Pi$ have potential $U$ satisfying Assumption \ref{ass:m sc and M smoothness}.
For a given preconditioner $L\in GL_{d}(\mathbb{R})$ for which Assumptions \ref{ass:explicit eigenvalue} and \ref{ass:explicit eigenvector} hold, the condition number
after preconditioning satisfies
\[
\tilde{\kappa}\leq(1+\epsilon)^{2}\left(1+\delta\sqrt{\sum_{i=1}^{d}\sigma_{i}^{2}\sum_{i=1}^{d}\sigma_{i}^{-2}}\right)^{4}
\]
\end{thm}

Proof can be found in Appendix \ref{proof:kappa_L_bound_1}. Assumption \ref{ass:explicit eigenvalue} states that the eigenvalues of $\nabla^2U(x)$ do not change much over $\mathbb{R}^d$. Assumption \ref{ass:explicit eigenvector} implies that  $v_i(x)^Tu_i \geq 1 - \delta$ for $i \in [d]$, $x, y \in \mathbb{R}^d$ and $v_i(x)^Tu_j \leq \delta$ for $i \in [d]$, $x, y \in \mathbb{R}^d$ where $i \neq j$, meaning that the eigenvectors also do not change much (see Appendix \ref{proof:v_control} for details).  In the Gaussian case $\epsilon = \delta = 0$ and so when $L=\Sigma_\pi^{-1/2}$ the bound becomes $\tilde{\kappa} \leq 1$, as expected.

\begin{rem}
    In the case of the additive Hessian, the eigenvalue stability inequality implies that
$\lambda_{i}(A)-\|B(x)\|\leq\lambda_{i}(A+B(x))\leq\lambda_{i}(A)+\|B(x)\|$
for all $i\in[d]$. Therefore choosing $L = A ^ {1/2}$ gives us that
\[
1-\frac{\|B(x)\|}{\lambda_d(A)} \leq \frac{\lambda_i(x)}{\sigma_i^2} \leq 1 + \frac{\|B(x)\|}{\lambda_d(A)}
\]
for all $i\in[d]$. So if $\|B(x)\|$ is small, $L = A ^ {1/2}$ will yield a smaller $\epsilon$. If $\|B(x)\|$ is large but $B(x)$ does not exhibit large variations we can simply restate the Hessian as $\nabla^2U(x)=\tilde{A}+\tilde{B}(x)$ where $\tilde{A}:=A + B(x^*)$ and $\tilde{B}(x):=B(x) - B(x^*)$ where $x^*\in\mathbb{R}^d$ is some point in the state space, and use $L = \tilde{A}^{1/2}$.
\end{rem}

\begin{rem}
Note that the quantity $\sqrt{\sum_{i}\sigma_{i}^{2}\sum_{i}\sigma_{i}^{-2}}=\sqrt{\text{Tr}(LL^{T})\text{Tr}((LL^{T})^{-1})}$ is upper bounded by $d\sqrt{\kappa(LL^T)}$. It can be viewed as an alternative `condition number' of $LL^T$, similar to that proposed in \citet{langmore:20}.
\end{rem}

\begin{rem}
The fact that $\sqrt{\sum_{i}\sigma_{i}^{2}\sum_{i}\sigma_{i}^{-2}}$
multiplies $\delta$ shows that when sampling from highly anisotropic distributions the penalty for misaligned eigenvectors of $LL^T$ relative to $\nabla^2 U$ is larger. As an example take $\Pi=\mathcal{N}(0,\Sigma_\pi)$ where $\Sigma_\pi\in\mathbb{R}^{2 \times 2}$ has eigendecomposition $\Sigma_\pi = Q_\pi D_\pi Q_\pi^T$ with $D_\pi=\text{diag}\{\lambda_1,\lambda_2\}$. We assume that $\Sigma_\pi$ is not a multiple of the identity. Construct the preconditioner $L = Q_\pi G D_\pi^{-1/2} G^T Q_\pi^T$ with the correct eigenvalues but whose eigenvectors have been perturbed by an orthogonal matrix $G$ from those of $\Sigma_\pi$ by the angle $\arccos(1-\delta)$. It can be shown that the coefficient of $\delta^4$ in $\tilde{\kappa}$ is $(1/4) \times (l - 2) ^ 2$ where $l := \lambda_1\lambda_2^{-1}+\lambda_1^{-1}\lambda_2$. So the more anisotropic $\Sigma_\pi$ is the more we are punished for having misaligned eigenvectors, as stated in the remark above. Figure \ref{fig:punished_misalignment} illustrates this fact. Each plot contains two contours: a blue one representing $\mathcal{N}(0, \Sigma_\pi)$ and an orange one representing $\mathcal{N}(0, (LL^T)^{-1})$. In both cases $G$ perturbs the eigenvectors by $\pi / 4$, with the angle between the semi-major axes of the contours shown in the red arrows. In the first case we have $(\lambda_1, \lambda_2) = (2, 1)$, in the second we have $(\lambda_1, \lambda_2) = (50, 1)$ engendering a far smaller `overlap' than in the first. The fact that $\tilde{\kappa} = O(\delta^4)$ also shows that the $\delta$ dependency in Theorem \ref{thm:kappa_L_bound_1} is tight.
\end{rem}

\begin{figure}[h]
    \centering
    \includegraphics[scale = 0.15]{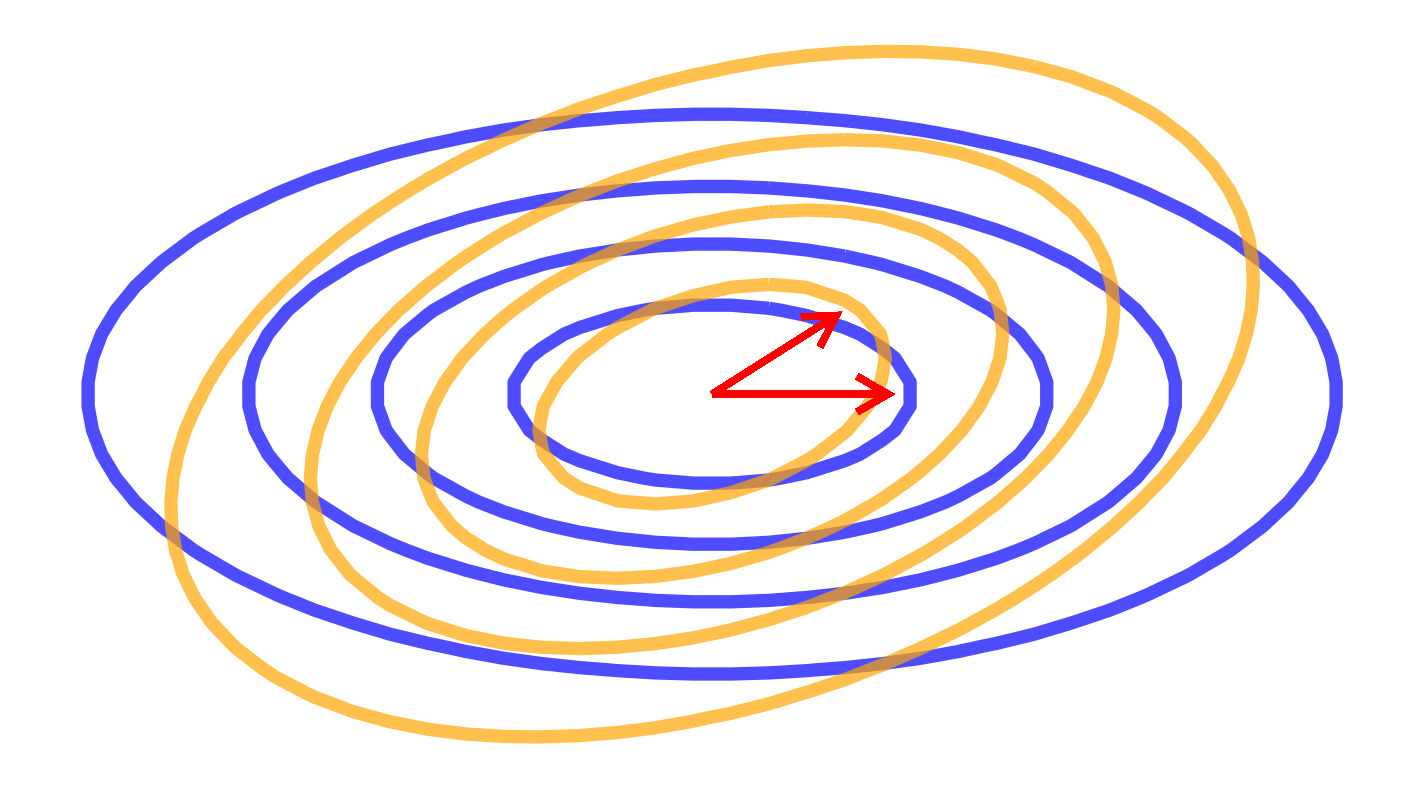}\includegraphics[scale = 0.15]{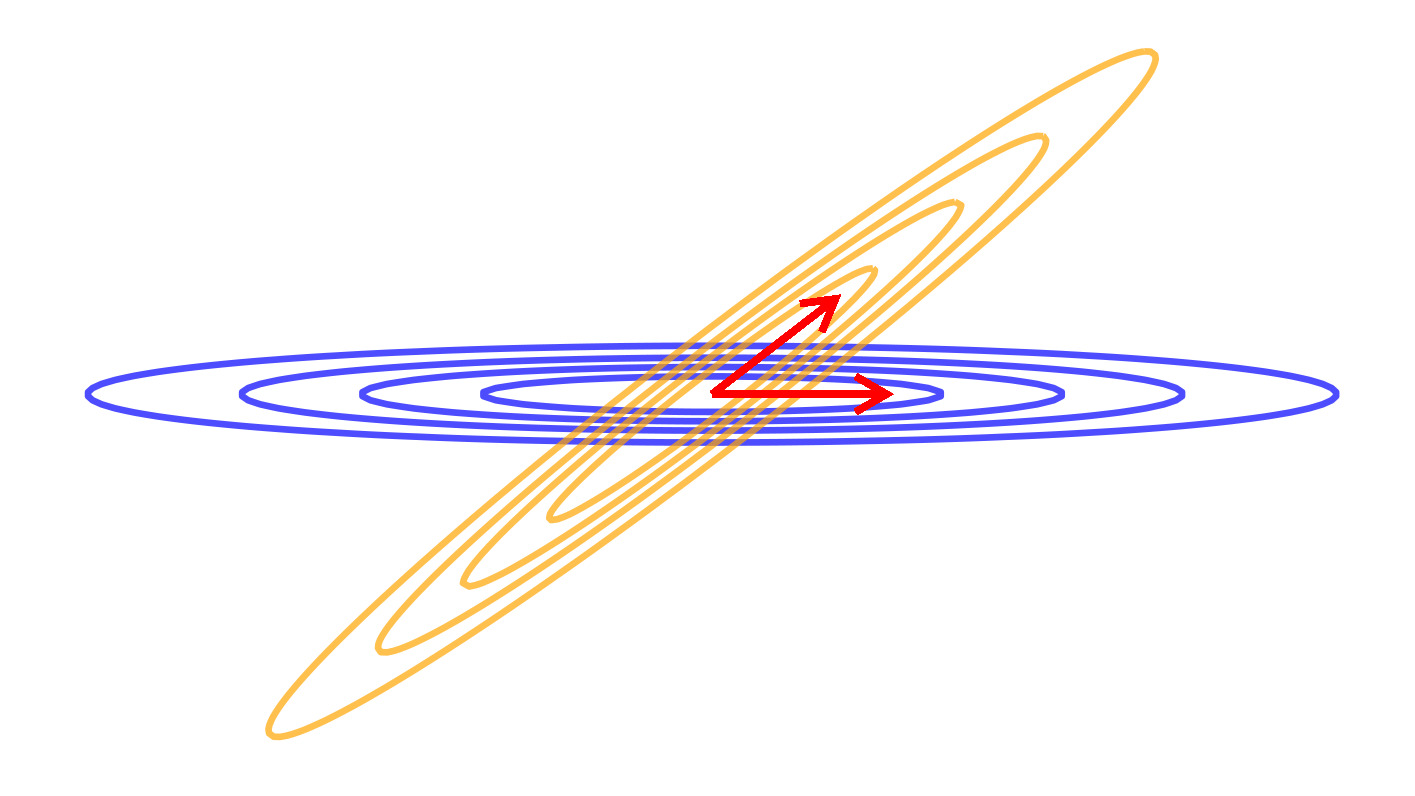}
    \caption{Two pairs of contour plots, each representing $\mathcal{N}(0, \Sigma_\pi)$ (blue) and $\mathcal{N}(0, (LL^T)^{-1})$ (orange). The angle between the semi-major axes (red) is the same in either case, but the preconditioner in the right hand plot is worse due to the anisotropy of $\mathcal{N}(0, \Sigma_\pi)$.}
    \label{fig:punished_misalignment}
\end{figure}

Assumption \ref{ass:explicit eigenvalue} and Assumption \ref{ass:explicit eigenvector} require knowledge of each individual eigenvalue and eigenvector of $\nabla^2U$ across the entire state space, which may not always be available. We next provide more easily verifiable assumptions under which a similar result holds.  This is achieved using results from matrix perturbation theory.  The eigenvalues of $\nabla^2 U$ can be controlled using only knowledge of the spectral norm by Weyl's inequality (see Assumption \ref{ass:hessian localisation}).  Similarly, using the Davis--Kahan theorem (e.g. \cite{yu:15}) eigenvectors can be controlled by the spectral norm provided that an `eigengap' condition holds (see Assumption \ref{ass:eigengap}).  We therefore provide a second result under Assumptions \ref{ass:hessian localisation} and \ref{ass:eigengap} below.

\begin{ass}\label{ass:hessian localisation} There exists an $\epsilon \geq 0$ such that $\|\nabla^2 U(x) - LL^T\| \leq \sigma_d^2\epsilon$ for all $x \in \mathbb{R}^d$.
\end{ass}

\begin{ass}\label{ass:eigengap} That $\gamma > 0$ where
    \[
    \gamma := \inf_{\substack{i, j \in [d] \\ |i - j| = 1}} |\sigma_i^2 - \sigma_j^2|
    \]
    is the \emph{eigengap} of $LL^T$.
\end{ass}

\begin{thm}\label{thm:kappa_L_bound_2}
Let $\Pi$ have potential $U$ satisfying Assumption \ref{ass:m sc and M smoothness}.
For a given preconditioner $L\in GL_{d}(\mathbb{R})$ for which Assumptions \ref{ass:hessian localisation} and \ref{ass:eigengap} hold, the condition number
after preconditioning satisfies
    \[
\tilde{\kappa}\leq(1+\epsilon)^{2}\left(1+\delta\sqrt{\sum_{i=1}^{d}\sigma_{i}^{2}\sum_{i=1}^{d}\sigma_{i}^{-2}}\right)^{4}
\]
where $\delta := 1 - (1 - 2\gamma ^ {-1}\sigma_d ^ {-2}\epsilon) ^ 2$.
\end{thm}

For a proof see Appendix \ref{proof:kappa_L_bound_2}.  Based on this result, it might be tempting to arbitrarily increase the eigengap $\gamma$ or the least eigenvalue $\sigma_d$ of $L$. Note, however, that this will also cause $\epsilon$ to increase.

\begin{rem}
    Appendix \ref{proof:A3 implies A1} shows that Assumption \ref{ass:hessian localisation} implies Assumption \ref{ass:explicit eigenvalue}.  Similarly Assumption \ref{ass:hessian localisation} and Assumption \ref{ass:eigengap} combined imply Assumption \ref{ass:explicit eigenvector} \citep[Corollary~1]{yu:15}.  If the norm $\|\nabla^2U(x)-LL^T\|$ is difficult to compute the Frobenius norm can be used to form the upper bound since $\|\nabla^2U(x)-LL^T\| \leq \|\nabla^2U(x)-LL^T\|_F$.
\end{rem}

The eigengap condition Assumption \ref{ass:eigengap} is always satisfiable as we are free to choose $L$.  It may not, however, be desirable.  We therefore present a final bound on $\tilde{\kappa}$ that requires only Assumption \ref{ass:hessian localisation}.

\begin{thm}\label{thm:kappa_L_bound_3}
    Let $\Pi$ have potential $U$ satisfying Assumption \ref{ass:m sc and M smoothness}. For a given preconditioner $L\in GL_d(\mathbb{R})$ for which Assumption \ref{ass:hessian localisation} holds, the condition number after preconditioning satisfies
    \[
    \tilde{\kappa} \leq (1 + \epsilon)\left(1 + \frac{\sigma_1^2}{m}\epsilon\right)
    \]
\end{thm}

For a proof see Appendix \ref{proof:kappa_L_bound_3}.  In the next two subsections we consider some popular choices of linear preconditioner that can be shown to reduce the condition number when the Hessian of $\Pi$ is of additive form, as described above.




\subsubsection{The Fisher matrix}
\label{subsubsec:fisher_matrix}

\citet{titsias:23}
suggests using $L\propto\mathcal{I}^{1/2}$ where $\mathcal{I}:=\mathbb{E}_{\Pi}[\nabla U(x)\nabla U(x)^{T}]$
is called the \emph{Fisher matrix}. This choice of $L$ maximises
the expected squared jump distance of the unadjusted Langevin algorithm,
which is simply the Euler-Maruyama discretisation of the Langevin diffusion given in equation (\ref{eq:langevin_diffn}). A straightforward integration by parts shows that $\mathcal{I}$ can also be written $\mathbb{E}_{\Pi}[\nabla^{2}U(x)]$, which
highlights its relationship with $\nabla^2 U$. Proposition \ref{prop:fisher_bound} shows that if a choice of $L$ satisfying Assumption \ref{ass:hessian localisation} exists, then the alternative choice of preconditioner $\mathcal{I}^{1/2}$ will also be valid.

\begin{prop}\label{prop:fisher_bound}
Let $\Pi$ have potential $U\in C^2$ and assume there exists a preconditioner $L\in GL_d(\mathbb{R})$ satisfying Assumption \ref{ass:hessian localisation} for some $\epsilon > 0$. Then $\|\nabla^2U(x) - \mathcal{I}\|\leq 2\sigma_d^2\epsilon$ for all $x \in \mathbb{R}^d$ where $\sigma_d^2$ is the least eigenvalue of $LL^T$.
\end{prop}
\begin{proof}
The assumption has that $\|\nabla^2U(x)-LL^T\|\leq \sigma_d^2\epsilon$ for all $x \in \mathbb{R}^d$. Therefore
\begin{align*}
\|\nabla^2U(x)-\mathcal{I}\| &= \|\nabla^2U(x)-LL^T + LL^T - \mathcal{I}\|\\
&\leq \sigma_d^2\epsilon + \left\|\mathbb{E}_\pi[\nabla^2U(X)-LL^T]\right\|\\
&\leq \sigma_d^2\epsilon + \mathbb{E}_\pi[\|\nabla^2U(X)-LL^T\|]\\
&\leq 2\sigma_d^2\epsilon
\end{align*}
where the penultimate line is due to Jensen's inequality. \qedsymbol
\end{proof}
\begin{cor}\label{cor:fisher_bound}
    Consider $\Pi$ satisfying Assumption \ref{ass:m sc and M smoothness} with a potential $U\in C^2$ and an $L\in GL_d(\mathbb{R})$ satisfying Assumption \ref{ass:hessian localisation} for some $\epsilon > 0$. Then choosing the preconditioner $\mathcal{I}^{1/2}$ gives
    \[
    \tilde{\kappa} \leq (1 + 2\epsilon)\left(1 + \frac{\sigma_1^2}{m}2\epsilon\right)
    \]
    where $\sigma_1^2$ is the greatest eigenvalue of $LL^T$.
\end{cor}
The proof of the above is a direct application of Theorem \ref{thm:kappa_L_bound_3}. 
Another commonly used choice of linear preconditioner is to set $L=\nabla^{2}U(x^{*})^{1/2}$
where $x^{*}$ is chosen to make $\nabla^{2}U(x^{*})$ sufficiently
representative of $\nabla^{2}U(x)$ everywhere (e.g. the mode). This choice avoids computing expectations with respect to $\Pi$. We test this choice of preconditioner empirically in Section \ref{subsec:experiment_multiplicative}.

\subsubsection{The target covariance}
\label{subsubsec:target_covariance}

Preconditioning with an estimate of the target covariance
is a popular strategy. To give some intuition for why this strategy is sensible we study
the spectral gap of the \emph{Ornstein--Uhlenbeck} (O-U) process, and show how this changes
under preconditioning. The preconditioned O-U process is an instance
of the preconditioned Langevin diffusion which is driven by the stochastic differential equation:
\begin{equation}\label{preconditioned langevin}
    dX_{t}=-\frac{1}{2}(LL^T)^{-1}\nabla U(X_t)dt+L^{-1}dB_{t},
\end{equation}
where $(B_{t})_{t\geq 0}$ is a Brownian motion. The preconditioned
O-U process is found by setting $\Pi=\mathcal{N}(0,\Sigma_{\pi})$, giving
\[
dX_{t}=-\frac{1}{2}(LL^T)^{-1}\Sigma_{\pi}^{-1}X_{t}dt+L^{-1}dB_{t}.
\]
Since our practical goal is to simulate this process on a computer,
we must take care when interpreting results about its continuous time
formulation. We therefore subject ourselves to the condition $\left|\det(-L^{-1}L^{-T}\Sigma_{\pi}^{-1})\right|=1$, which precludes the choice $L' := s^{-1}L$ for increasingly large $s>0$, which would arbitrarily increase the rate of convergence of
the continuous time process but destabilise the discretised process for any fixed numerical integrator step-size.  The below result is most likely well-known, but we were unable to find a reference, and so provide it here for completeness.

\begin{prop}\label{prop:optimal_OU}
The preconditioner $L=\Sigma_{\pi}^{-1/2}$ maximises the spectral
gap of the preconditioned O-U process subject to $\left|\det(-L^{-1}L^{-T}\Sigma_{\pi}^{-1})\right|=1$.
\end{prop}

Proof of Proposition \ref{prop:optimal_OU} can be found in Appendix \ref{proof:optimal_OU}. \citet{alrachid:18} also identify the spectral gap of the preconditioned
O-U process, although since they focus on the continuous time dynamics,
they do not identify an optimal preconditioner (with respect to the
spectral gap).

In Section \ref{subsec:additive_hessian} we characterised the performance
of a preconditioner $L$ by how much the Hessian $\nabla^{2}U(x)$
varies around $LL^{T}$. Here we do the same for $L=\Sigma_{\pi}^{-1/2}$.

\begin{prop}\label{prop:localise_covariance}
Let $\Pi$ be a measure with potential $U$, covariance $\Sigma_{\pi}\in\mathbb{R}^{d\times d}$,
mode $x^{*}\in\mathbb{R}^{d}$, and expectation $\mu_{\pi}\in\mathbb{R}^{d}$.
Assume that there exist positive definite matrices $\Delta_{+},\Delta_{-}\in\mathbb{R}^{d\times d}$
such that $\Delta_{-}\preceq\nabla^{2}U(x)\preceq\Delta_{+}$ for all $x \in \mathbb{R}^d$ and that $1-(x^*-\mu_\pi)^T\Delta_+(x^*-\mu_\pi)>0$. Then
$P_{-}\preceq\Sigma_{\pi}^{-1}\preceq P_{+}$ where
\begin{align*}
P_{+} & =c^{-1}\left(\mathbf{I}_d+\left(1-\text{\emph{Tr}}(D_+)\right)^{-1}D_+\right)\Delta_+\\
P_{-} & =c\left(\mathbf{I}_d+\left(1-\text{\emph{Tr}}(D_-)\right)^{-1}D_{-}\right)\Delta_-
\end{align*}
with $D_\pm = \Delta_\pm(x^*-\mu_\pi)(x^*-\mu_\pi)^T$ and $c:=\sqrt{\det\Delta_{-}{\det\Delta_{+}^{-1}}}\leq1$, and in addition
$$
\|\nabla^2U(x)-\Sigma_\pi^{-1}\|\leq\max\{\|\Delta_+ - P_-\|,\|P_+ - \Delta_-\|\}
$$
for all $x\in\mathbb{R}^d$.
\end{prop}

For a proof see Appendix \ref{proof:localise_covariance}.  Proposition \ref{prop:localise_covariance} allows us to localise the covariance in terms
of the parameters of the target distribution. 

The intuition from Section \ref{subsec:additive_hessian} suggests that Hessians which exhibit small variations across the state space are preconditionable. In this scenario we would have $\Delta_+\approx\Delta_-$ and hence that $c \approx 1$. Proposition \ref{prop:localise_covariance} then suggests that so long as the distance between the mean and the mode is not too great, $\|\nabla^2U(x)-\Sigma_\pi^{-1}\|$ is small. In summary, if $\Pi$ is preconditionable, and if the mean is close to the mode, then preconditioning with $L=\Sigma_\pi^{-1/2}$ is sensible.

Recall that a potential with additive Hessian satisfies $\nabla^2U(x) = A+B(x)$ where $A,B(x)\in\mathbb{R}^{d\times d}$ are symmetric.
In this case $\Delta_{-}$ and $\Delta_{+}$ will be generated by
variations in $B(x)$. Therefore, given Proposition \ref{prop:localise_covariance},
a tighter localisation of $B(x)$ gives a tighter localisation of
$\nabla^{2}U(x)$ around the inverse covariance, leading to the following result.

\begin{cor}\label{cor:localise_covariance_additive}
Let $\Pi$ be a measure with potential $U$, covariance $\Sigma_{\pi}\in\mathbb{R}^{d\times d}$,
mode $x^{*}\in\mathbb{R}^{d}$, and expectation $\mu_{\pi}\in\mathbb{R}^{d}$.
If the Hessian of $U$ is of the form $\nabla^{2}U(x)=A+B(x)$
with $\|B(x)\|\leq\epsilon$ and $\epsilon\mathbf{I}_{d}\prec A$ for some
$\epsilon>0$, then if $1-(x^*-\mu_\pi)^T(A+\epsilon\mathbf{I}_d)(x^*-\mu_\pi)>0$ it follows that
\[
\|\nabla^{2}U(x)-\Sigma_{\pi}^{-1}\|\leq(c^{-1}+1)\epsilon+(c^{-1}-1)\|A\|+\max\{\|\tilde{P}_{-}\|,\|\tilde{P}_{+}\|\},
\]
where
\begin{align*}
\tilde{P}_{+} & =c^{-1}\left(1-\text{\emph{Tr}}(\tilde{D}_+)\right)^{-1}\tilde{D}_+(A+\epsilon\mathbf{I}_{d}),\\
\tilde{P}_{-} & =c\left(1-\text{\emph{Tr}}(\tilde{D}_-)\right)^{-1}\tilde{D}_-(A-\epsilon\mathbf{I}_{d}),
\end{align*}
with $\tilde{D}_\pm = (A \pm \epsilon \mathbf{I}_d)(x^*-\mu_\pi)(x^*-\mu_\pi)^T$ and c:=$\sqrt{{\det(A-\epsilon\mathbf{I}_{d})}{\det(A+\epsilon\mathbf{I}_{d})^{-1}}}\leq1$.
\end{cor}

A proof can be found
in Appendix \ref{proof:localise_covariance_additive}. In Section \ref{subsec:experiment_additive} we look at the difference in performance between the preconditioners $L = A^{1/2}$, $L = \Sigma_\pi^{-1/2}$ and $L = \mathbf{I}_d$.

\subsection{Linear preconditioning for multiplicative Hessians}\label{subsec:multiplicative_hessian}

A multiplicative Hessian has the form
\begin{equation}\label{multiplicative hessian definition}
    \nabla^{2}U(x)=Z^{T}\Lambda(x)Z
\end{equation}
where $Z\in\mathbb{R}^{n\times d}$ for $n\geq d$ is a matrix
whose rows are usually the rows of some dataset and $\Lambda(x)\in\mathbb{R}^{n\times n}$.
An example of a model with a multiplicative Hessian is as follows
\begin{equation}\label{Bayesian_regression}
    \pi(\theta)\propto\exp\left(-\sum_{k=1}^{n}l_{y_{k}}(z_{k}^{T}\theta)-\frac{\lambda}{2}(\theta-\mu)^{T}Z^T \Lambda Z(\theta-\mu)\right)
\end{equation}
where $\{(y_{k},z_{k})\}_{k=1}^{n}$ are observations with $y_{k}\in\mathbb{R}$ and $z_{k}\in\mathbb{R}^{d}$ for
$k\in[n]$, $Z$ is a matrix with element in row $i$ and column $j$ equal to the $j$th element of $z_i$, and $\Lambda \in\mathbb{R}^{n \times n}$ is
positive definite.  Here $\ell_{y_k}$ denotes some loss associated with observation $k$. In the case $\ell_{y_k}$ is a negative log-likelihood, equation \eqref{Bayesian_regression} therefore describes the posterior associated with a typical generalised linear model using the generalised g-prior of \citet{hanson:14}.


Without further assumptions we have the following:
\begin{prop}\label{prop:ostrowski_rectangular}
A measure $\Pi$ whose potential $U$ has a multiplicative Hessian satisfies
\[
\frac{\sup_{x\in\mathbb{R}^d}\lambda_{n-d+1}(\Lambda(x))}{\inf_{x\in\mathbb{R}^d}\lambda_d(\Lambda(x)\kappa(Z^TZ)} \leq \kappa \leq \kappa(Z^TZ)\frac{\sup_{x\in\mathbb{R}^d}\lambda_1(\Lambda(x))}{\inf_{x\in\mathbb{R}^d}\lambda_d(\Lambda(x))}
\]
\end{prop}

The proof in Appendix \ref{proof:ostrowski_rectangular} relies on an extension of  Ostrowski's theorem to rectangular matrices \citep[Theorem~3.2]{higham:98a}. In many cases $\Lambda(x)$ will be diagonal and each eigenvalue $\lambda_i(\Lambda(x))$ will range between the same possible values $c$ and $C$ (this is the case for binary logistic regression using the g-prior, for example). In this instance a more precise statement about $\kappa$ can straightforwardly be made.

\begin{ass}\label{ass:Lambda bound} The Hessian of $\Pi$ is of multiplicative form with $\Lambda(x)$ diagonal, and there exists $c, C > 0$ such that $\sup_{x\in\mathbb{R}^d}\lambda_i(\Lambda(x)) = C$ and $\inf_{x\in\mathbb{R}^d}\lambda_i(\Lambda(x)) = c$ for all $i \in [n]$.
\end{ass}

\begin{prop}\label{prop:condition_number_multiplicative}
    A measure $\Pi$ for which Assumption \ref{ass:Lambda bound}  holds has a condition number
    \[
    \kappa=\frac{C}{c}\kappa(Z^TZ)
    \]
\end{prop}

\subsubsection{The choice $L = (Z^TZ)^{1/2}$}

A natural choice of preconditioner here is $L = (Z^TZ)^{1/2}$, which is proportional to that suggested by \citet[Section~6.2]{dalalyan:17}. For this choice we have the following result.

\begin{prop} \label{prop:mult_dalalyan}
    Consider a measure $\Pi$ whose potential has a multiplicative Hessian. Then preconditioning with $L = (Z^TZ)^{1/2}$ gives
    \[
    \frac{\sup_{x\in\mathbb{R}^d}\lambda_{n-d+1}(\Lambda(x))}{\inf_{x\in\mathbb{R}^d}\lambda_d(\Lambda(x))} \leq \tilde{\kappa} \leq\frac{\sup_{x\in\mathbb{R}^d}\lambda_1(\Lambda(x))}{\inf_{x\in\mathbb{R}^d}\lambda_d(\Lambda(x))} 
    \]
\end{prop}

For a proof see Appendix \ref{proof:mult_dalalyan}. Under Assumption \ref{ass:Lambda bound} the upper and lower bounds of Proposition \ref{prop:mult_dalalyan} are equal, giving the following.

\begin{cor}\label{cor:mult_dalalyan}
    Consider a measure $\Pi$ for which Assumption \ref{ass:Lambda bound} holds. Choosing $L = (Z^TZ)^{1/2}$ gives
    \[
    \tilde{\kappa} = \frac{C}{c}
    \]
\end{cor}

The condition number is therefore clearly reduced under preconditioning with $L = (Z^TZ)^{1/2}$, since $\kappa/\tilde{\kappa} = \kappa(Z^TZ)\geq 1$.  The level of reduction will be determined by $\kappa(Z^TZ)$, which characterises how far from orthogonal $Z^TZ$ is.  In the context of regression or classification problems $\kappa(Z^TZ)$ will be smallest under an orthogonal design and larger when there is more collinearity among features or when different features have very different variances.

\subsubsection{The QR decomposition}
\label{subsubsec:QR}

A popular strategy for regression and classification models in which $Z$ is a design matrix is to perform a (reduced) QR decomposition, setting $Z = QR$ where $Q \in \mathbb{R}^{n \times d}$ is orthogonal and $R \in \mathbb{R}^{d \times d}$ is upper triangular (e.g. \cite[Section~1.2]{stanuserguide}).  In this case the preconditioner is $L = R$, from which it readily follows by setting $Z = Q$ in Proposition \ref{prop:ostrowski_rectangular} that
$$
\tilde{\kappa} \leq \frac{\sup_{x\in\mathbb{R}^d}\lambda_1(\Lambda(x))}{\inf_{x\in\mathbb{R}^d}\lambda_d(\Lambda(x))}.
$$
The QR strategy therefore gives the same upper bound as the choice $(Z^TZ)^{1/2}$ of the previous subsection.

\subsubsection{The Hessian at the mode}

Another natural choice of preconditioner is $L = \nabla^2 U(x^*)^{1/2} = (Z^T\Lambda(x^*)Z)^{1/2}$ for some $x^* \in \mathbb{R}^d$. For instance $x^*$ could be the mode of $\Pi$. Here we have the following

\begin{prop}\label{prop:mult_Hessian_at_point}
    Consider a measure $\Pi$ whose potential has a multiplicative Hessian. Then choosing $L = (Z^T\Lambda(x^*)Z)^{1/2}$ gives
    \[
    \tilde{\kappa} \leq \frac{\sup_{x\in\mathbb{R}^d}\lambda_1(\Lambda(x^*)^{-1/2}\Lambda(x)\Lambda(x^*)^{-1/2})}{\inf_{x\in\mathbb{R}^d}\lambda_d(\Lambda(x^*)^{-1/2}\Lambda(x)\Lambda(x^*)^{-1/2})}  \leq \left(\frac{\sup_{x\in\mathbb{R}^d}\lambda_1(\Lambda(x))}{\inf_{x\in\mathbb{R}^d}\lambda_d(\Lambda(x))}\right)^2
    \]
\end{prop}

A proof can be found in Appendix \ref{proof:mult_Hessian_at_point}.  This strategy might be preferable to the previous two options in the setting where $\Lambda(x)$ still has some variation between eigenvalues, but when it does not change much for different values of $x$, meaning it is an almost constant matrix.  In that case $\Lambda(x^*)^{-1/2}\Lambda(x)\Lambda(x^*)^{-1/2}$ should be close to the identity matrix, meaning the condition number after preconditioning with the Hessian at the mode will be $\approx 1$, whereas $\sup_{x\in\mathbb{R}^d}\lambda_1(\Lambda(x))/\inf_{x\in\mathbb{R}^d}\lambda_d(\Lambda(x))$ might still be much larger than 1.  We compare the Hessian at the mode to the choice $L = (Z^TZ)^{1/2}$ empirically in Section \ref{subsec:experiment_multiplicative}.

\subsection{Tight condition number dependence of the spectral gap of random walk Metropolis}
\label{subsec:spectral_gap}

\citet{andrieu:22} give upper and lower bounds on the spectral gap of the random walk Metropolis that are tight in the dimension, but only the lower bound is explicit in its dependence on the condition number. Here we show that if some additional conditions are imposed on $U$ akin to an additive Hessian structure then the dependence on $\kappa$ can also be made explicit in the upper bound.

\begin{ass}\label{ass:Hessian only localisation} $\Pi$ satisfies Assumption \ref{ass:m sc and M smoothness} and in addition there is an $\epsilon > 0$ such that $\|\nabla^2U(x)-\nabla^2 U(y)\|\leq m\epsilon$ for all $x,y\in\mathbb{R}^d$.  
\end{ass} 

\begin{rem}
    Assumption \ref{ass:Hessian only localisation} will hold when $\nabla^2U(x) = A + B(x)$ with $\|B(x)\|$ suitably small, since $\sup_x\|B(x)\|<m\epsilon/2 \implies \|\nabla^2U(x) - \nabla^2U(y)\| \leq 2\sup_x\|B(x)\| \leq m\epsilon$ as required.  Assumption \ref{ass:Hessian only localisation} is therefore simply one way to formalize the notion of `additive' Hessian structure.
\end{rem}  

\begin{thm}\label{RWM_mixing_bounds}
    Let $\Pi$ have potential $U$ satisfying Assumption \ref{ass:m sc and M smoothness} with condition number $\kappa=M/m\geq 1$. If $\Pi$ also satisfies Assumption \ref{ass:Hessian only localisation} then the spectral gap $\gamma_{\kappa}$ of the random walk Metropolis with proposal variance $\sigma^2I$ such that $\sigma ^ 2 := \xi/(Md)$ for any $\xi > 0$ satisfies
    \[
    C\xi\exp(-2\xi) \cdot \frac{1}{\kappa}\cdot \frac{1}{d}\leq\gamma_{\kappa}\leq(1 + 2\epsilon)\frac{\xi}{2} \cdot \frac{1}{\kappa} \cdot \frac{1}{d}
    \]
    where $C = 1.972 \times 10 ^ {-4}$
\end{thm}

For a proof see Appendix \ref{proof:RWM_mixing_bounds}.  Both bounds presented above are $O(\kappa^{-1})$, which implies that the relaxation time $1/\gamma_{\kappa}$ is precisely linear in $\kappa$. The lower bound is as originally presented by \citet[Theorem~1]{andrieu:22}, and the choice $\sigma^2 = \xi/(Md)$ is as recommended in that work to ensure tight $O(d^{-1})$ dependence. As the authors remark, the constant $C$ can possibly be made a few orders of magnitude larger.

Theorem \ref{RWM_mixing_bounds} above can be combined with the condition number reduction results of Section \ref{subsec:additive_hessian} to guarantee that under the additive Hessian assumption the spectral gap increases under appropriate linear preconditioning, as shown in the below corollary.

\begin{cor}\label{cor:improved_gap}
    Let $\Pi$ have potential $U$ satisfying Assumption \ref{ass:m sc and M smoothness} with condition number $\kappa=M/m\geq 1$. Assume that $\Pi$ also satisfies Assumption \ref{ass:Hessian only localisation} with constant $\epsilon'>0$. Using a preconditioner $L\in GL_d(\mathbb{R})$ satisfying Assumption \ref{ass:hessian localisation} with constant $\epsilon>0$ ensures that the spectral gap $\gamma_{\kappa}$ of the RWM with proposal variance $\sigma^2I$ such that $\sigma ^ 2 := \xi/(Md)$ for any $\xi > 0$ increases under preconditioning whenever
    \[
        \kappa \geq \frac{1}{2}C ^ {-1}\exp(2\xi)(1+2\epsilon')(1+\epsilon)\left(1+\frac{\sigma_1(L)^2}{m}\epsilon\right)
    \]
    where $C = 1.972 \times 10 ^ {-4}$, and $\sigma_d(L)$ is the least singular value of $L$.
\end{cor}

Corollary \ref{cor:improved_gap} uses Theorem \ref{RWM_mixing_bounds} along with Theorem \ref{thm:kappa_L_bound_3}, which provides an upper bound on the condition number after preconditioning in the additive Hessian setting. This bound leads to a lower bound on the relaxation time, which can be compared with the upper bound from Theorem \ref{RWM_mixing_bounds} before preconditioning. This comparison establishes that preconditioning ensures an improvement in relaxation time, provided the initial condition number $\kappa$ is sufficiently large.  A proof can be found in \ref{proof:improved_gap}. Note that a similar result could be stated by applying the bounds of Theorem \ref{thm:kappa_L_bound_1} or Theorem \ref{thm:kappa_L_bound_2} in place of Theorem \ref{thm:kappa_L_bound_3}, which would then modify the necessary lower bound on $\kappa$ in the above.  We leave the details of this as a simple exercise for the reader.

\subsection{Counterproductive diagonal preconditioning}\label{subsubsec:diagonal_counterproductive}

The appeal of diagonal preconditioning is motivated
by the promise of reducing the condition number in $O(d)$ computational cost. The practice has a strong
tradition in numerical linear algebra, where diagonal preconditioning
is commonplace as discussed in Section \ref{subsec:optimisation}. Therefore choosing $L=\text{diag}(\hat{\Sigma}_{\pi})^{-1/2}$ for some estimate $\hat{\Sigma}_\pi$ of the target covariance has become a common practice since it is viewed as computationally cheap, and it is assumed that it will offer an improvement on no preconditioning at all. The developers of Stan, for instance, offer the option of diagonal preconditioning with the target covariance as using a diagonal mass matrix \citep{carpenter2017stan}.
Such an option is also offered in the TensorFlow Probability library
\citep{abadi:15}. 

In fact, we show here that there are examples of distributions for which diagonal preconditioning in this manner can actually \emph{increase} the condition number, and therefore be explicitly \emph{worse} than performing no pre-conditioning at all.  This phenomenon can be observed even when the target is Gaussian. Noting that $\tilde{\kappa}=\|\text{diag}(\Sigma_{\pi})^{1/2}\Sigma_{\pi}^{-1}\text{diag}(\Sigma_{\pi})^{1/2}\| \cdot \|\text{diag}(\Sigma_{\pi})^{-1/2}\Sigma_{\pi}\text{diag}(\Sigma_{\pi})^{-1/2}\|=\|C_{\pi}^{-1}\|\|C_{\pi}\|$
where $C_{\pi}\in\mathbb{R}^{d\times d}$ is the correlation matrix
associated with $\Sigma_{\pi}$, it suffices to find a
$\Sigma_{\pi}$ for which $\tilde{\kappa} = \kappa(C_\pi) > \kappa(\Sigma_\pi) = \kappa$. The matrix below is such an example.
\begin{equation}\label{gaussian_covariance}
    \Sigma_{\pi}=
\begin{pmatrix}21.5 & 5.7 & 18.7 & 4.5 & 6.9\\
\star & 2.0 & 4.9 & 1.2 & 2.1\\
\star & \star & 16.3 & 3.9 & 5.7\\
\star & \star & \star & 1.4 & 1.4\\
\star & \star & \star & \star & 2.9
\end{pmatrix}
\implies\kappa\approx4.4\times10^{3},\tilde{\kappa}\approx8.1\times10^{3}
\end{equation}
In the above we have truncated the entries to a single decimal place: see Appendix \ref{appendix:B} for the full matrix. The condition number increases by a substantial amount, even though
we have perfect knowledge of the target covariance. See Section \ref{subsec:experiment_counterproductive} for an empirical analysis of the random walk Metropolis on a $\mathcal{N}(0, \Sigma_\pi)$ target (with $\Sigma_\pi$ as above) after diagonal and dense preconditioning.

In a practical
scenario we would have to expend computational effort to construct $\text{diag}(\hat{\Sigma}_\pi)^{-1/2}$ and so for a target such as the one described here,
this effort would be wasted as it actually reduces sample quality. In general we can conjecture that targets with covariance matrices whose
associated correlations are far from being diagonally dominant will
be least amenable to diagonal preconditioning.

\section{Experiments}\label{sec:experiments}

\subsection{Counterproductive Diagonal Preconditioning}\label{subsec:experiment_counterproductive}

This experiment illustrates the phenomenon described in section \ref{subsubsec:diagonal_counterproductive}: namely that there exist Gaussian targets in the form $\mathcal{N}(0, \Sigma_\pi)$ such that preconditioning with $L = \text{diag}(\Sigma_\pi) ^ {-1/2}$ \emph{increases} the condition number.

We compare the performance of three RWM algorithms: one with no preconditioning, one with dense preconditioning ($L = \Sigma_\pi^{-1/2}$), and one with diagonal preconditioning ($L = \text{diag}(\Sigma_\pi) ^ {-1/2}$). Each chain targets $\mathcal{N}(0, \Sigma_\pi)$ with $\Sigma_\pi$ as in (\ref{gaussian_covariance}) and is initialised at equilibrium. Proposals take the standard RWM form $X'=X+\sigma L ^ {-1} \xi$ with $\xi\sim\mathcal{N}(0, \mathbf{I}_5)$ and $\sigma = 2.38 / \sqrt{d}$ as recommended by \citet{roberts:01}. We run each chain from each algorithm 100 times at 10,000 iterations per chain. For each chain we compute the \emph{effective sample size} (ESS) in each dimension using the \texttt{effectiveSize} function from the \texttt{coda} package \citep{plummer:06}. The ESS is an estimator which measures the amount of \emph{independent} samples one would need to achieve an empirical average with equivalent variance to the one computed from the chain.


\begin{figure}
    \centering
\includegraphics[scale = 0.25]{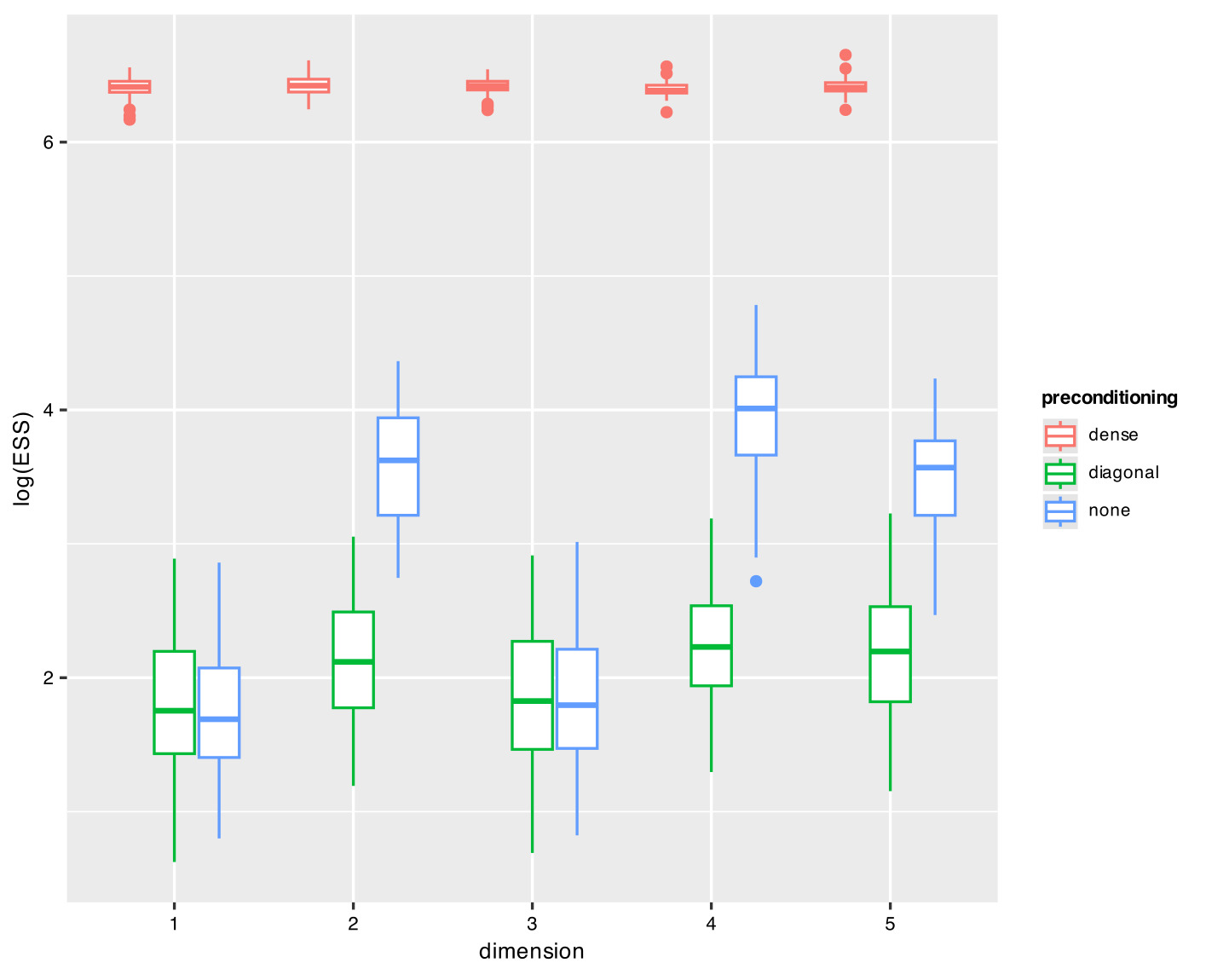}
    \caption{$\log\text{ESS}$ of 100 runs at 10,000 iterations per run of RWM with dense, diagonal, and zero preconditioning. Each algorithm is started in equilibrium and targets $\mathcal{N}(0, \Sigma_\pi)$ with $\Sigma_\pi$ as in (\ref{gaussian_covariance}).}
    \label{fig:counterproductive_diagonal}
\end{figure}

As can be seen in Figure \ref{fig:counterproductive_diagonal} the ESS's of the RWM chain with no preconditioning are clearly larger in dimensions 2, 4, and 5 from the diagonally preconditioned chain. This is despite the diagonal preconditioner being formed with perfect knowledge of the target covariance. As is expected the dense preconditioner performs the best since the target effectively becomes a standard Gaussian.

\subsection{Preconditioning the additive Hessian}\label{subsec:experiment_additive}

One probabilistic model with additive Hessian structure
is a Bayesian regression with hyperbolic prior. It is well
known that the Laplace prior $\beta\mapsto2^{-1}\lambda\exp(-\lambda\|\beta\|_{1})$
for $\beta\in\mathbb{R}^{d}$ and $\lambda>0$ imposes the same sparsity
in the maximum a posteriori estimates as would regularisation with
the LASSO \citep{tibshirani:96} since it concentrates sharply
around $\beta=0$. More generally, priors with exponential tails can be motivated by results concerning the contraction of regression parameter posteriors around their true values. As demonstrated in the discussion following Theorem 8 in \cite{castillo:15}, heavy-tailed priors such as the Laplace distribution achieve good rates of contraction. Simply using a Laplace prior would violate the $M$-smoothness assumption due to the behaviour at $\beta=0$. The hyperbolic
prior, however, is a smooth distribution with exponential tails, and so we will use this as a prior for the regression parameters $\beta \in \mathbb{R}^d$. We assume that $Y=X\beta+\epsilon$
where $\epsilon\sim\mathcal{N}(0,\sigma^{2}\mathbf{I}_{n})$ for $\sigma^{2}>0$
known and $n>d$. We also assume that the columns of $X\in\mathbb{R}^{n\times d}$
are standardised to have variance 1. Because of this it is reasonable to use the same scale
in each dimension of the prior. The resulting posterior has a potential
of the form
\begin{equation}\label{eq:regression_hyperbolic_Hessian}
        U(\beta) =\frac{1}{2\sigma^{2}}\|Y-X\beta\|^{2}+\lambda\sum_{i=1}^{d}\sqrt{1+\beta_{i}^{2}},
\end{equation}
which implies
\begin{equation} \label{eq:grad_hyperbolic_hessian}
\nabla^{2}U(\beta) =\frac{1}{\sigma^{2}}X^{T}X+\lambda D(\beta),
\end{equation}
where $D(\beta)=\text{diag}\{(1+\beta_{i}^{2})^{-3/2}:i\in[d]\}$.
The Hessian is therefore additive with $A=\sigma^{-2}X^{T}X$ and $B(\beta)=\lambda D(\beta)$.

Equations \eqref{eq:regression_hyperbolic_Hessian}-\eqref{eq:grad_hyperbolic_hessian} show that the posterior is a well-conditioned distribution, meaning that $U$ satisfies Assumption \ref{ass:m sc and M smoothness}
with $m=\sigma^{-2}\sigma_{d}(X^{T}X)$ and $M=\sigma^{-2}\|X^{T}X\|+\lambda$.
Preconditioning with $L=(\sigma^{-2}X^{T}X)^{1/2}$ gives
$\tilde{\kappa}=1+\sigma_{d}(X^{T}X)^{-1}\lambda\sigma^{2}\leq\kappa$.
In this case the distance between $LL^{T}$ and the Hessian can be
bounded using $\|\nabla^{2}U(\beta)-LL^{T}\|\leq\lambda$, so we can also apply the results of Section \ref{subsubsec:target_covariance} to justify the use of the target covariance by setting $\epsilon=\lambda$ in Corollary \ref{cor:localise_covariance_additive}.
 These results imply that when $\lambda$ is small then preconditioning with either $L=(\sigma^{-2}X^{T}X)^{1/2}$
or $L=\Sigma_{\pi}^{-1/2}$ should
improve the efficiency of the sampler, so long as the distance between
the mean and the mode is not too large.

In the following experiment we run MALA chains on target distributions with $L=(\sigma^{-2}X^{T}X)^{1/2}$,
$L=\Sigma_{\pi}^{-1/2}$, and $L=\mathbf{I}_{d}$. We set
$d\in\{2,5,10,20,100\}$ and $n=\{1,5,20\}\times d$ for each value of $d$.
At each combination of $n$ and $d$ we run 15 chains for each preconditioner.
Each chain is composed by initialising at $\beta=(X^{T}X)^{-1}X^{T}Y$
and taking $10^{4}$ samples to equilibrate. We initialise the step size at $d^{-1/6}$ and adapt it
along the course of the chain seeking an optimal acceptance rate
of $0.574$ according to the results of \citet{roberts:01}. We then continue the
chain with preconditioning and a fixed step size of $d ^ {-1/6}$ for a further $10^{4}$ samples, over which
we measure the ESS of each dimension. To construct $L=\Sigma_{\pi}^{-1/2}$
we simply use the empirical covariance of the first $10^{4}$ samples.

For the model parameters we set $\lambda=\sqrt{n}/d$ using the lower
bound of \citet{castillo:15}. Every element of $X$ is an independent
standard normal random variable, and $Y$ is generated by sampling
$\beta_{0}$ from the prior and setting $Y=X\beta_{0}+\epsilon$ with
$\epsilon\sim\mathcal{N}(0,\mathbf{I}_{d})$, meaning $\sigma=1$.

The boxplots in Figure \ref{fig:covariance_preconditioning_additive_Hessian} show the median ESSs for each run. The figure demonstrates that in the $n/d \in \{5,20\}$ cases preconditioning with $L = \Sigma_\pi^{-1/2}$ is just as good as preconditioning with $L = A^{1/2}$ where $A=\sigma^{-1}X^TX$. For $n/d=1$ the results are mixed: for instance in the $(d, n) = (100, 100)$ configuration the first $10^4$ iterations of the MALA chain mixed poorly, offering a poor estimate of $\Sigma_\pi$. The performance suffered heavily if no preconditioner was applied.

\begin{figure}
    \centering
\includegraphics[scale = 0.27]{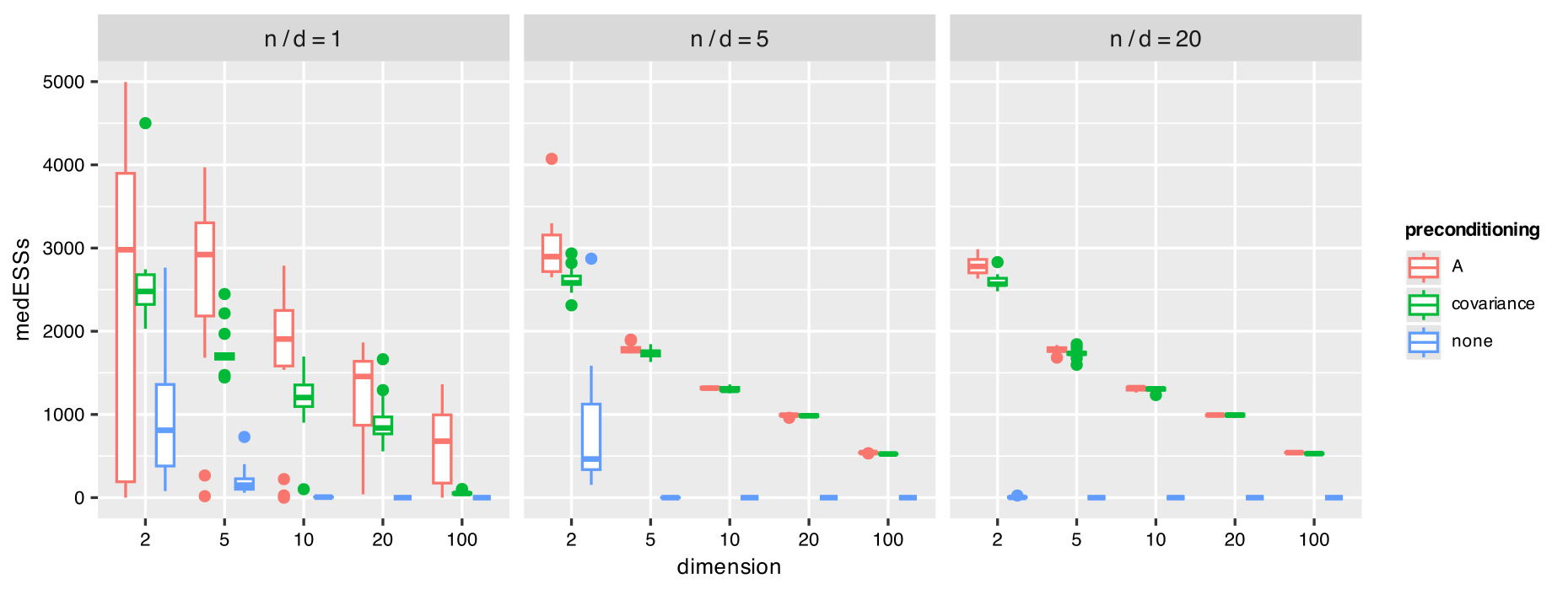}
    \caption{Boxplots of the medians of the ESSs across configurations of $(n, d)$ with different preconditioners on the Bayesian linear regression with a Hyperbolic prior. The leftmost boxplot in each grouping corresponds to preconditioning with $L = \sigma (X^TX)^{1/2}$ (`A' in the legend), the middle boxplot has $L = \Sigma_\pi^{-1/2}$ (`covariance' in the legend), the rightmost has $L=\mathbf{I}_d$ (`none' in the legend).}
    \label{fig:covariance_preconditioning_additive_Hessian}
\end{figure}

\subsection{Preconditioning the multiplicative Hessian}\label{subsec:experiment_multiplicative}

To study preconditioning under the multiplicative Hessian structure we consider
a Bayesian binomial regression with a generalised $g$-prior \citep[Section~2.1]{bové:10} \citep[Section~2.2]{held:17}.
The generalised $g$-prior is an extension of the classical
$g$-prior to generalised linear models that have dispersion parameters of the form
$\phi_{i}:=\phi w_{i}^{-1}$ for $i\in[n]$ and known weights
$w_{i}\in\mathbb{R}^{+}$. It is motivated by constructing an `imaginary
sample' of responses $y_{0}=h(0)\mathds{1}_{n}\in\mathbb{R}^{n}$
from a generalised linear model with inverse link function $h(.)$ and design matrix $X\in\mathbb{R}^{n\times d}$.
Assigning the parameter vector $\beta\in\mathbb{R}^{d}$ a flat prior,
it is observed that as $n\to\infty$ the posterior distribution of
$\beta$ in this construction tends to $\mathcal{N}_{d}(0,g\phi c(X^{T}WX)^{-1})$
where $W=\text{diag}\{w_{i}:i\in[n]\}$, $g$ and $\phi$ are hyperparameters,
and $c$ is a model-specific constant. See \citet[Section~2.1]{bové:10} for a more thorough
exposition.

We follow the
advice given by \citet[section~2.1]{bové:10} and \citet[section~2.2]{held:17} by setting $w_{i}=m_{i}$ for all $i\in[n]$. Using a logistic link gives
a posterior with potential
\begin{equation}\label{eqn:binomial_potential}
    U(\beta)=\sum_{i=1}^{n}\left(w_{i}\left((1-Y_{i})X_{i}^{T}\beta+\log(1+\exp(-X_{i}^{T}\beta))\right)\right)+(g\phi c)^{-1}\beta^{T}X^{T}WX\beta
\end{equation}
with Hessian $\nabla^{2}U(\beta)=X^{T}\Lambda(\beta)X$, where
\[
\Lambda(\beta):=W\text{diag}\{\exp(X_{i}^{T}\beta)(1+\exp(X_{i}^{T}\beta))^{-2}+(g\phi c)^{-1}:i\in[n]\}
\]
The potential $U$ therefore satisfies Assumption \ref{ass:m sc and M smoothness} with $M=(0.25+(g\phi c)^{-1})w_{\text{max}}\|X^{T}X\|$
and $m=(g\phi c)^{-1}w_{\text{min}}\sigma_{d}(X^{T}X)$, where $w_{\text{max}}:=\max_{i}w_{i}$
and $w_{\text{min}}:=\min_{i}w_{i}$. We
choose $g$ and $\phi$ such that $(g\phi c)^{-1}=\lambda n^{-1}$,
where $\lambda=0.01$.

We examine the effectiveness of preconditioning
with $L\in\{\Sigma_{\pi}^{-1/2},\mathcal{I}^{1/2},\nabla^{2}U(\beta^{*})^{1/2},\mathbf{I}_{d},(n^{-1}X^{T}X)^{1/2}\}$
where $\Sigma_{\pi}$ is the covariance of the posterior, $\mathcal{I}$
is the `Fisher matrix' of \citet{titsias:23}, $\beta^{*}$
is the mode, and $L=(n^{-1}X^{T}X)^{1/2}$ is the preconditioner used in \citet[Section~6.2]{dalalyan:17}. When $L\in\{\mathbf{I}_{d},(n^{-1}X^{T}X)^{1/2},\nabla^{2}U(\beta^{*})^{1/2}\}$
the condition numbers are given by
\begin{align*}
L=\mathbf{I}_{d} & \Rightarrow\tilde{\kappa}=\kappa=\frac{\frac{n}{4}+\lambda}{\lambda}\frac{w_{\text{max}}}{w_{\text{min}}}\kappa(X^{T}X)\\
L=(n^{-1}X^{T}X)^{\frac{1}{2}} & \Rightarrow\tilde{\kappa}=\frac{\frac{n}{4}+\lambda}{\lambda}\frac{w_{\text{max}}}{w_{\text{min}}}\\
L=\nabla^{2}U(\beta^{*})^{\frac{1}{2}} & \Rightarrow\tilde{\kappa}=\frac{\frac{n}{4}+\lambda}{\lambda}\frac{n\max_{i}p_{i}^{*}(1-p_{i}^{*})+\lambda}{n\min_{i}p_{i}^{*}(1-p_{i}^{*})+\lambda}
\end{align*}
where $p_{i}^{*}:=(1+\exp(-X_{i}^{T}\beta^{*}))^{-1}$. This suggests
that $L=\nabla^{2}U(\beta^{*})^{1/2}$ offers an increase
in efficiency over $L=(n^{-1}X^{T}X)^{1/2}$ for $w_{\text{max}}/w_{\text{min}}$
large.

\subsubsection{Experimental setup and results}
We run RWM chains with the preconditioners described above for $d\in\{2,5,10,20\}$
and $n=5d$. We generate the design matrix $X\in\mathbb{R}^{n\times d}$
with $X=G+M$ where $G_{ij}\sim\mathcal{N}(0,1)$ independently and
$M_{ij}=\mu$ for all $i\in[n],j\in[d]$. We set $\mu \in \{0,5,50,200\}$
to arbitrarily worsen the conditioning of the model, as it can be shown
that 
\[
\kappa(X^{T}X)\geq\frac{\sum_{k=1}^{n}(G_{k1}+\mu)^{2}}{\frac{1}{2}\sum_{k=1}^{n}(G_{k1}-G_{k2})^{2}}.
\]
We set $w_{i}=i^{2}$ for $i\in[n]$ and generate the responses
using $Y_{i}=S_{i}/w_{i}$ with $S_{i}\sim\text{Bin}(w_{i},(1+\exp(-X_{i}^{T}\beta_{0}))^{-1})$
for $\beta_{0}\sim\mathcal{N}(0,\mathbf{I}_{d})$. We use gradient descent on $U$ which we precondition with $L=(n^{-1}X^{T}X)^{1/2}$ to find the
mode $\beta^{*}$.

We approximate $\Sigma_{\pi}$ and $\mathcal{I}$ in two different ways.
We either construct them using ergodic averages generated by unpreconditioned
RWM for $10^{4}$ iterations, or we run an $L=\nabla^{2}U(\beta^{*})^{1/2}$
preconditioned RWM for $10^{5}$ iterations, from which we calculate
the same ergodic averages. At each combination of $d$ and $\mu$
we run 15 chains for each preconditioner. Each chain is composed by
initialising at $\beta\sim\mathcal{N}(0,(n^{-1}X^{T}X)^{-1})$ and
taking $10^{4}$ samples to equilibrate. In each of these initial chains we initialise the step
size at $2.38/d^{1/2}$ and adapt it along the course of the
trajectory seeking an optimal acceptance rate of $0.234$ according
to the results of \citet{roberts:01}. We then continue the chain
with preconditioning and a fixed step size of $2.38/d^{1/2}$ for a further $10^{4}$ samples, over which we
measure the ESS of each dimension. The median
ESSs  in the $\mu \in \{0,200\}$ cases are plotted in Figure \ref{fig:binomial_regression}.

\begin{figure}
\centering
\includegraphics[scale = 0.3]{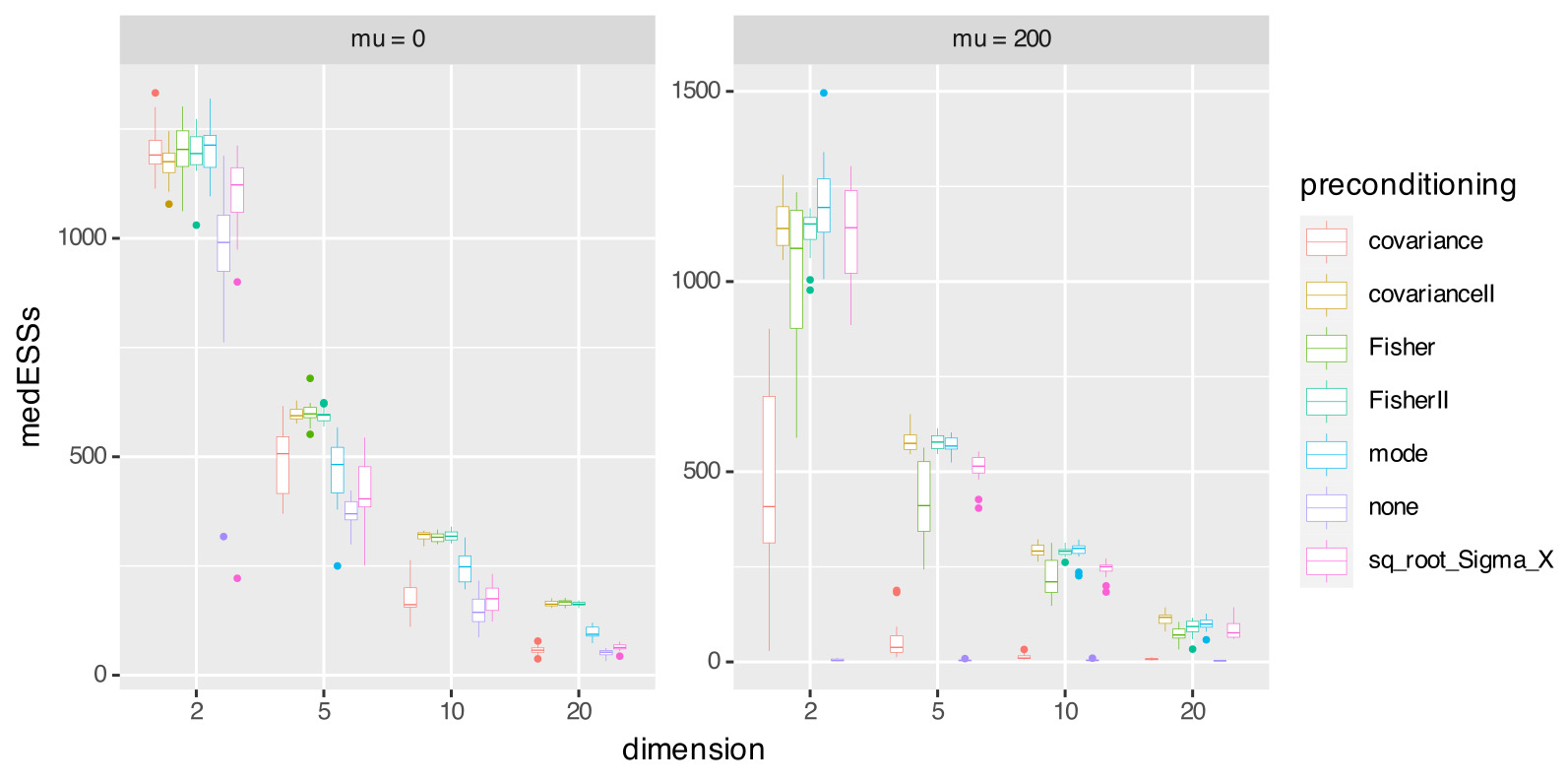}
    \caption{Boxplots of the logarithms of the medians of the ESSs across combinations of $(d, \mu)$. The ESSs are taken from RWM runs on a binomial regression target with the generalised $g$-prior. `covariance' and `covarianceII' correspond to runs preconditioned with $L = \Sigma^{-1/2}$ where $\Sigma_\pi$ is estimated over ${10}^4$ and ${10}^5$ runs respectively. `Fisher' and `FisherII' correspond to runs preconditioned with $L = \mathbb{E}_\pi[\nabla^2U(\beta)]^{1/2}$ where $\mathbb{E}_\pi[\nabla^2U(\beta)]$ is estimated over ${10}^4$ and ${10}^5$ runs respectively. `mode' refers to runs preconditioned with $L = \nabla^2U(\beta^*)^{1/2}$ where $\beta^*$ is an estimate of the mode found using preconditioned gradient descent. `sq\textunderscore root\textunderscore Sigma\textunderscore X' corresponds to runs preconditioned with $L = (n^{-1}X^TX)^{1/2}$.}
    \label{fig:binomial_regression}
\end{figure}

The preconditioning strategies are detailed in the figure caption. `covariance' and `covarianceII' refer to the runs preconditioned with $L = \Sigma_\pi^{-1/2}$ where the covariance is estimated over $10^4$ and $10^5$ samples respectively. The same is the case for `Fisher' and `FisherII'. `sq\textunderscore root\textunderscore Sigma\textunderscore X' refers to the runs made with $L=(n^{-1}X^TX)^{1/2}$. `mode' refers to preconditioning with $L=\nabla^{2}U(\beta^{*})^{1/2}$.

As predicted, preconditioning with the Hessian at the mode does offer a benefit over preconditioning with $L=(n^{-1}X^TX)^{1/2}$. Preconditioning with the covariance when it is estimated over a larger, better quality sample (`covarianceII') is one of the best performing strategies, whereas preconditioning with the covariance estimated over the smaller sample (`covariance') suffers with dimension and ill-conditioning of the model. This is clearly due to the reduction in quality of the covariance estimate. This disparity in performance is contrasted with the difference between the `Fisher' and `FisherII' cases, which is very slight. 

\subsection{Reducing computational cost in Hamiltonian Monte Carlo}\label{subsec:experiment_hmc}

In the specific context of Hamiltonian Monte Carlo the cost of a large condition number can be somewhat compensated for by choosing the trajectory length of the Hamiltonian flow at each iteration appropriately.  In the case of a Gaussian target with diagonal covariance \cite{bou2017randomized} show that if this trajectory length is randomly sampled as $T \sim \text{Exp}(\lambda^{-1})$ then $\lambda \propto \sqrt{1/m}$ achieves optimal $\lambda$-adjusted mean square displacement between successive states of the Markov chain.  Given that the numerical integrator step-size must be set $\propto \sqrt{1/M}$ for numerical stability then this implies that the number of leap-frog integration steps employed per iteration must be chosen as $L = T/\epsilon \propto \sqrt{\kappa}$.  An appropriately chosen mass matrix should therefore reduce the amount of computation needed per iteration, and therefore the overall cost of the sampler presuming that the per iteration quality of samples is comparable.

We illustrate this phenomenon empirically in two settings in Figure \ref{fig:HMC_comparisons}. The left hand plot shows the results of running HMC samplers on a 50 dimensional Gaussian $N(0, \Sigma_\pi)$ target where $\Sigma_\pi$ is a diagonal matrix. These HMC chains are started in equilibrium and run for ${10}^4$ iterations. In the unpreconditioned case, we set $\lambda$ according to the guidance of \citet[equation (32)]{bou2017randomized}. The step size is adapted to achieve an average acceptance rate of 0.65 during the first half of each simulation, and effective sample sizes are then measured over the second half. The preconditioner used is $L = \Sigma_\pi^{-1/2}$. The right hand plot shows the results of running HMC samplers on a 50 dimensional Bayesian logistic posterior with a $g$-prior. The potential of this posterior is equivalent to the binomial regression potential in (\ref{eqn:binomial_potential}) with $w_i = 1$ for all $i \in [n]$. For each target we randomly sample $X \in\mathbb{R}^{n \times d}$ in order to control the condition number of the resulting posterior. In each case we set $n = d$. Each chain is initialised at the mode and run for ${10}^4$ iterations. In the unpreconditioned case we set $\lambda = \sqrt{1 / m^*}$ where $m^*$ is the least eigenvalue of $\nabla^2U(x^*)$. The step size is again adapted to achieve an average acceptance rate of 0.65 over the first half of each simulation, and effective sample sizes are then measured over the second half. The preconditioner used is $L = (n^{-1}X^TX)^{-1/2}$, as recommended by \cite{dalalyan:17}.

In each plot within Figure \ref{fig:HMC_comparisons} the vertical axis shows the median number of leapfrog steps (across the dimensions of the Markov chain) needed to achieve an independent sample as a multiple of the preconditioned performance, hence the number of leapfrog steps in the preconditioned cases is always 1. The horizontal axis shows the condition number of the target. An approximately square root relation between these two quantities can be seen in both plots. It is clear that in both cases the preconditioner vastly reduces the amount of computation needed to achieve an independent sample.  Additional plots in Figure \ref{fig:HMC_comparisons} below replacing number of leapfrog steps with wall clock time on the vertical axes show analogous behaviour.

\begin{figure}
    \centering
    \includegraphics[scale = 0.25]{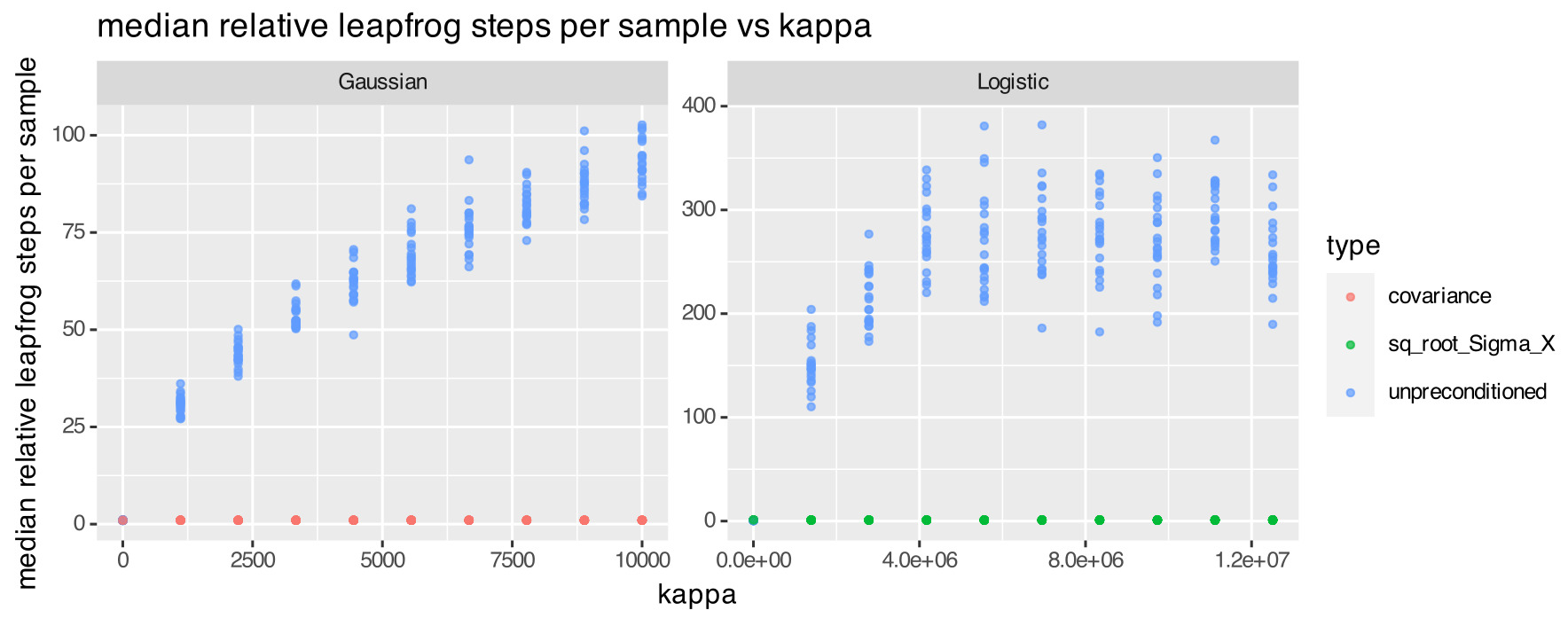}
    \includegraphics[scale = 0.25]{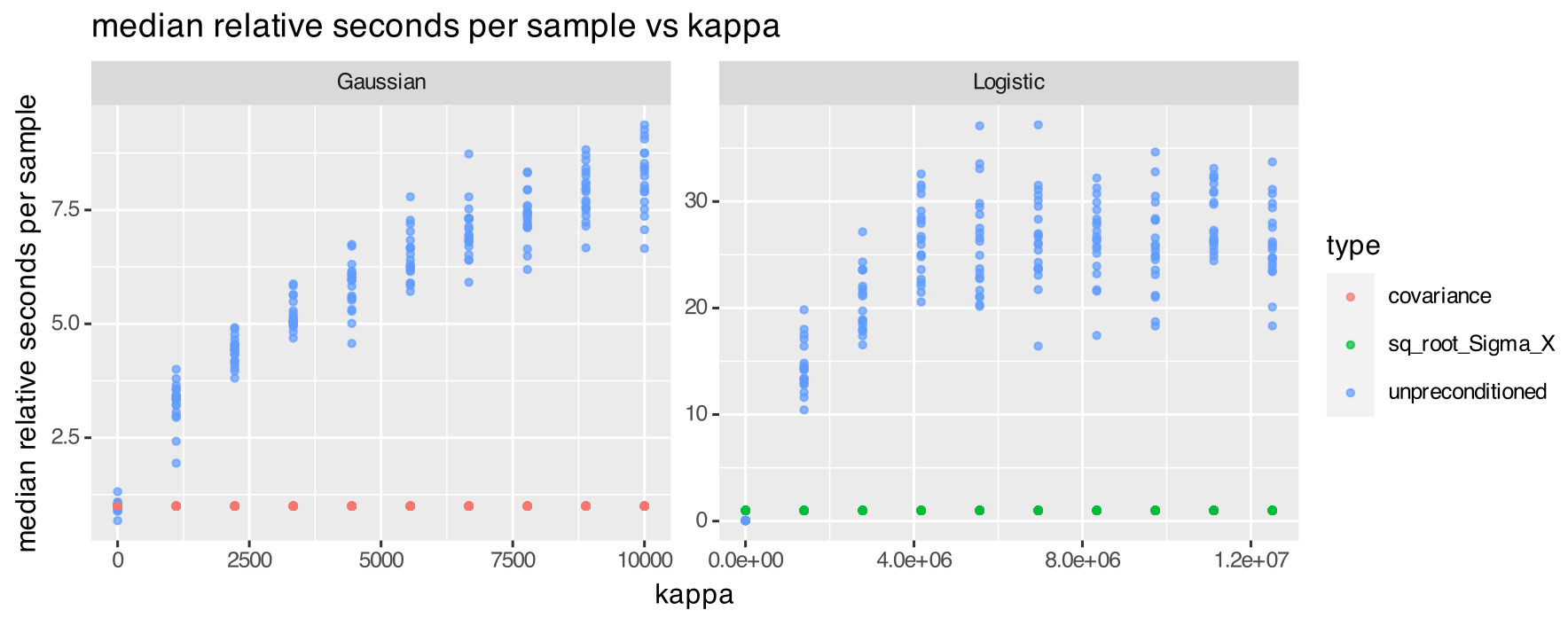}
    \caption{Left plots: preconditioned and non-preconditioned HMC chains for a Gaussian target with diagonal covariance. Right plots: Bayesian logistic regression posterior with a $g$-prior. The top vertical axes denote the median number of leapfrog steps (across the dimensions of the Markov chain) needed to achieve an independent sample, the bottom vertical axes denote the median time (across dimensions) needed to achieve an independent sample. The horizontal axes denote the condition number of the target in all plots. `covariance' refers to runs preconditioned with $L = \Sigma_\pi^{-1/2}$ and `sq\_root\_Sigma\_X' refers to runs preconditioned with $L = (n^{-1}X^TX)^{1/2}$.}
    \label{fig:HMC_comparisons}
\end{figure}

\section{Discussion}




\subsection{Advice to practitioners}


Our first cautionary tale for practitioners is that linear preconditioning will not always be beneficial, there are some target distributions for which it will not improve the quality of MCMC sampling, and can even reduce it.  We have shown, however, that under conditions that intuitively insist that the contours of $\pi(x)$ look `not too far from elliptical', as characterised by the additive and multiplicative Hessian structures, then appropriate linear preconditioning should improve a sampling algorithm, and can do so considerably if the target distribution is far from isotropic.

The most generically applicable preconditioner is certainly the inverse square root of the empirical covariance. This can always be computed using either pilot runs or online within an adaptive MCMC algorithm, even when the Hessian of $U(x)$ is not available, provided that the target has finite variance. We have shown that if enough samples are used in its estimation and the target satisfies the conditions we impose then this strategy works well both in theory and in practice. Alternative choices can also improve sampling, but we did not observe any cases in which an alternative strategy is clearly better than the target covariance, in fact most preconditioners offered very similar levels of improvement in our empirical study.  One exception to this is in the case of binomial regression, where the Hessian at the mode strategy slightly outperformed the choice $(Z^TZ)^{1/2}$, something which is supported by our theory in Section \ref{subsec:multiplicative_hessian}.

We have shown that, contrary to the received wisdom, diagonal preconditioning can in fact be damaging to the quality of MCMC samples, even in cases when a dense preconditioner can significantly improve the algorithm.  This seems to occur when the correlations present in the target distribution are particularly strong. If it is suspected that the target distribution is of this form it is important to compare diagonal preconditioning to none at all, as well as considering if the computational budget is available to perform dense preconditioning, which might be much more beneficial.

Finally we have highlighted that appropriate preconditioning can significantly reduce computational cost in Hamiltonian Monte Carlo, by reducing the number of numerical integrator steps needed at each iteration to solve Hamilton's equations. It is usually considered that preconditioning is a very expensive operation, but in fact it is the trade off between the likely $O(d^2)$ cost of full linear preconditioning and the typically at least $O(\sqrt{\kappa})$ cost of not doing so which must be evaluated for a given problem. We have given two clear examples in Section \ref{subsec:experiment_hmc} in which preconditioning can improve sample quality despite the additional matrix algebra that might be needed for its implementation.

\subsection{Extensions}

Our work offers many interesting extensions, which we now discuss in turn.

\subsubsection{Alternative condition numbers and refinements}

Alternative, problem specific condition numbers have been defined by various parties. For instance where $\Sigma_\pi$ is the covariance of $\Pi$ with spectrum $\sigma^2_1 \geq \sigma^2_2 \geq ... \geq \sigma^2_d$ \citet{langmore:20} suggest using
\[
\left(\sum_{i=1}^d \left(\frac{\sigma_1}{\sigma_i}\right)^4\right)^\frac{1}{4}
\]
as a condition number. Under specifications on the step size of the algorithm, such a quantity is shown to be proportional to the number of leapfrog steps needed to achieve a stable acceptance rate in HMC.

The condition number as defined in (\ref{condition_number}) encodes the difficulty of sampling from $\Pi$, but it does not capture additional information we might have about $\Pi$ which might ameliorate the sampling efficiency. For instance, if we knew that there existed positive definite $A_-,A_+\in\mathbb{R}^{d \times d}$ such that 
$$
A_- \preceq \nabla^2U(x)\preceq A_+,
$$
then we could precondition with $L = A_-^{1/2}$ achieving $\tilde{\kappa} = \lambda_1(A_-^{-1}A_+)$ over $\kappa = \lambda_1(A_+) / \lambda_d(A_-)$. Therefore defining the condition number as $\lambda_1(A_-^{-1}A_+)$ encodes the difficulty of sampling \emph{given all the information at hand}. See, for instance, \citet[Section~2.3]{safaryan:21}, \citet[Definition~2.9]{saumard:14} or \citet[Definition~1]{hillion:19} for similarly motivated definitions.

\subsubsection{Nonlinear preconditioning} \label{subsec:nonlinear}

A natural extension is to broaden the class of preconditioners to include nonlinear transformations. At present nonlinear preconditioning can be seen in the form of normalizing flows \citep{gabrie:22, hoffman:19} and measure transport \citep{parno:18}. Less computationally intensive transformations are considered by \citet{johnson:12} and \citet{yang:22} in order to sample from heavy-tailed distributions.

We note that to identify a transformation $g:\mathbb{R}^d\to\mathbb{R}^d$ such that the pushforward of $\Pi$ under $g$ has condition number 1 is to solve the equation
\[
U(x)+\log|\det J(g(x))|=\frac{1}{2}\|g(x)\|^{2}
\]
where $J(g(x))$ is the Jacobian of $g$ at $x\in\mathbb{R}^d$. This is an instance of the \emph{Monge-Ampère} equation, which is well studied in optimal transport \citep[remark~2.25]{peyré:20}. Solvers of the Monge-Ampère exist in the literature, see \citet{benamou:14, benamou:16b}. Contextualising the existing analysis of the Monge-Ampère and its solvers within MCMC is a potentially fruitful line of inquiry.

There exist classes of algorithms that are equivalent to transforming existing sampling algorithms under nonlinear transformations. These include the \emph{Riemannian manifold} algorithms of \citet{girolami:11} (see also \citet{patterson:13,lan:15,livingstone:21}) and the algorithms derived from \emph{mirror descent} \citep{nemirovski:83} such as those seen in \citet{hsieh:18, zhang:20wasserstein, chewi:20}. That the algorithms derived from mirror descent are equivalent to a nonlinearly preconditioned sampling scheme is evident in their construction. For the Riemannian manifold samplers, one can show, for instance, that the Langevin diffusion
\[
dY_{t}=\frac{1}{2}\nabla\log\tilde{\pi}(Y_{t})dt+dB_{t}
\]
under diffeomorphism $f(Y)=X$ transforms into the following SDE
\begin{equation}\label{manifold langevin}
    \begin{split}
        dX_{t} & =\frac{1}{2}G(X_{t})^{-1}\nabla\log\pi(X_t)dt+\Gamma(X_{t})dt+G(X_{t})^{-\frac{1}{2}}dB_{t}\\
    \Gamma_{i}(X_{t}) & =\frac{1}{2}\sum_{j=1}^{d}\frac{\partial}{\partial x_{j}}\left(G(X_{t})_{ij}^{-1}\right)
    \end{split}
\end{equation}
with $G(x)^{-1}=J(g(x))^{-1}J(g(x))^{-T}$ where $g$ is the inverse of $f$ and $\pi(x)=\tilde{\pi}(y)|\det J(f(y))^{-1}|$, see \citet{zhang:23transport} for a formal statement and proof. \citet{xifara:14,livingstone:14} show that the SDE in \ref{manifold langevin} is the Langevin diffusion on the Riemannian manifold with metric $G(x)\in\mathbb{R}^{d \times d}$, and therefore the same diffusion underlying the Riemannian manifold MALA algorithm of \citet{girolami:11} is equivalent to an instance of nonlinear preconditioning. One can make a similar equivalence in the case of Riemannian manifold HMC, whereby we make a nonlinear transformation to the momentum variable used in (\ref{eq:hamiltonian_dynamics}), see \citet{hoffman:19} for an explanation.

These equivalences provide motivation for further study. For instance if one can identify a $g$ such that the metric $J(g(x))^{-1}J(g(x))^{-T}$ matches that used by \citet{girolami:11} one can bypass the computationally costly operations inherent in the Riemannian manifold methods. One can also evaluate the benefits of using Riemannian schemes with arbitrary metrics by evaluating the change in the condition number under transformations which achieve those metrics.

\subsubsection{Beyond well-conditioned distributions}

The condition number as defined in (\ref{condition_number}) is restrictive in the class of models it applies to, namely distributions satisfying Assumption \ref{ass:m sc and M smoothness}. Where $\Pi$ satisfies $M$-smoothness and a \emph{Poincaré inequality}: for all $f\in L^1(\Pi)$
\[
\text{Var}_\pi(f)\leq C_\mathsf{PI}\mathbb{E}_\pi[\|\nabla f\|^2]
\]
with constant $C_\mathsf{PI}\geq0$ \citet[footnote, page~3]{zhang:23improved} define it as $\kappa:=C_{\mathsf{PI}}M$. They are motivated by its presence in the mixing time bounds they derive for the unadjusted Langevin sampler. An application of the Brascamp-Lieb inequality shows that $C_\mathsf{PI}=m^{-1}$ in the case that $\Pi$ also has an $m$-strongly convex potential. \citet{chen:23} also derive mixing time bounds under a more general constraint than $m$-strong convexity. One could alternatively use the quantities involved in their constraints and therefore the mixing time bounds to redefine the condition number.

\citet{altmeyer:22} constructs a \emph{surrogate posterior} whose potential satisfies Assumption \ref{ass:m sc and M smoothness} and coincides with the potential of the target posterior on a region in which the target concentrates. Under assumptions, they provide polynomial time mixing bounds for unadjusted Langevin Monte Carlo using the fact that the chain will stay in the aforementioned region for exponentially long with high probability. The ability to identify such behaviour allows one to quantify the conditioning of a posterior whose potential violates Assumption \ref{ass:m sc and M smoothness}.


\subsubsection*{Acknowledgements}

MH is funded by an EPSRC DTP. SL acknowledges support from EPSRC grant EP/V055380/1. The authors thank Dootika Vats, Sam Power, Giacomo Zanella and Max Goplerud for useful discussions.


\newpage

\appendix
\section{}
\label{app:theorem}



\subsection{Proof of Proposition \ref{prop:hard_target}}\label{proof:hard_target}

The Hessian of the model with potential
\[
U(x,y)=\frac{m-M}{2}\left(\cos x+\cos y\right)+\frac{M+m}{2}\left(\frac{x^{2}}{2}+\frac{y^{2}}{2}\right)
\]
is in the form $\nabla^{2}U(x,y)=\text{diag}\{f(x),f(y)\}$
where $f(x):=(1/2)(M - m)\cos x+(1/2)(M + m)\in[m,M]$. As detailed
in Proposition \ref{prop:symmetry_invariance}, the condition number is ignorant
as to whether the preconditioner $L$ is symmetric or not, so we
assume it is. Therefore we can perform an eigendecomposition $L=QDQ^{T}$
where $D=\text{diag}\{\lambda_{1},\lambda_{2}\}$ is the matrix of
eigenvalues (not necessarily ordered) and, since we are in two dimensions,
$Q$ can be represented as the two dimensional Givens matrix 
\[
Q=\begin{pmatrix}\cos\theta & -\sin\theta\\
\sin\theta & \cos\theta
\end{pmatrix}.
\]
The matrix enclosed by the first operator norm in (\ref{conditioned_number}) has trace and determinant
\begin{align*}
\text{Tr}(x,y) & :=\text{Tr}(L^{-T}\nabla^{2}U(x,y)L^{-1})=c^{2}(\lambda_{1}^{-2}f(x)+\lambda_{2}^{-2}f(y))+s^{2}(\lambda_{2}^{-2}f(x)+\lambda_{1}^{-2}f(y))\\
\text{Det}(x,y) & :=\text{Det}(L^{-T}\nabla^{2}U(x,y)L^{-1})=\lambda_{1}^{-2}\lambda_{2}^{-2}f(x)f(y)
\end{align*}
where we have abbreviated $c:=\cos\theta$, $s:=\sin\theta$ for notational
simplicity. The matrix enclosed by the second operator norm in (\ref{conditioned_number}) has trace and determinant
\begin{align*}
\text{Tr}^{*}(x^{*},y^{*}) & :=\text{Tr}(L\nabla^{2}U(x^{*},y^{*})^{-1}L^{T})=c^{2}(\lambda_{1}^{2}f(x^{*})^{-1}+\lambda_{2}^{2}f(y^{*})^{-1})+s^{2}(\lambda_{2}^{-2}f(x^{*})^{-1}+\lambda_{1}^{-2}f(x^{*})^{-1})\\
\text{Det}^{*}(x^{*},y^{*}) & :=\text{Det}(L\nabla^{2}U(x^{*},y^{*})^{-1}L^{T})=\lambda_{1}^{2}\lambda_{2}^{2}f(x^{*})^{-1}f(y^{*})^{-1}
\end{align*}
Using the fact that the operator norm of a positive definite matrix
is simply the largest eigenvalue, we are able to lower bound
\[
\tilde{\kappa}\geq\frac{1}{2}\left(\text{Tr}(x,y)+\sqrt{\text{Tr}(x,y)^{2}-4\text{Det}(x,y)}\right)\frac{1}{2}\left(\text{Tr}^{*}(x^{*},y^{*})+\sqrt{\text{Tr}^{*}(x^{*},y^{*})^{2}-4\text{Det}^{*}(x^{*},y^{*})}\right)
\]
Choosing $(x,y)$ such that $f(x)=f(y)=M$ and $(x^{*},y^{*})$ such
that $f(x^{*})=f(y^{*})=m$ we have
\begin{align*}
\tilde{\kappa} & \geq\frac{1}{2}\left((\lambda_{1}^{-2}+\lambda_{2}^{-2})M+\left|\lambda_{1}^{-2}-\lambda_{2}^{-2}\right|M\right)\frac{1}{2}\left((\lambda_{1}^{2}+\lambda_{2}^{2})m^{-1}+\left|\lambda_{1}^{2}-\lambda_{2}^{2}\right|m^{-1}\right)\\
 & =\max\{\lambda_{1}^{-2},\lambda_{2}^{-2}\}\max\{\lambda_{1}^{2},\lambda_{2}^{2}\}\frac{M}{m}\\
 & =\kappa(LL^{T})\kappa
\end{align*}
Therefore $\tilde{\kappa}>\kappa$ for non-orthogonal $L$.

\subsection{Proof of Proposition \ref{prop:symmetry_invariance}}\label{proof:symmetry_invariance}

For a given preconditioner $L\in GL_d(\mathbb{R})$ we define $\tilde{L} \in GL_d(\mathbb{R})$ to be the positive definite matrix that has eigenvalues equal to the singular
values of $L$ and eigenvectors equal to the right singular vectors of $L$. Recall the definition of the condition number after preconditioning
with $L$: $\tilde{\kappa}=\sup_{x}\|L^{-T}\nabla^{2}U(x)L^{-1}\|\sup_{x}\|L\nabla^{2}U(x)^{-1}L^{T}\|$.
We inspect the norm inside the first supremum after computing the
singular value decomposition of $L=U\Sigma V^{T}$ where $U$ and $V$ are orthogonal
and $\Sigma\in\mathbb{R}^{d\times d}$ is a diagonal matrix whose
diagonal elements are the singular values of $L$.
\begin{align*}
\|L^{-T}\nabla^{2}U(x)L^{-1}\| & =\|U\Sigma^{-1}V^{T}\nabla^{2}U(x)V^{T}\Sigma^{-1}U\|\\
 & =\|\Sigma^{-1}V^{T}\nabla^{2}U(x)V^{T}\Sigma^{-1}\|\\
 & =\|V\Sigma^{-1}V^{T}\nabla^{2}U(x)V^{T}\Sigma^{-1}V\|\\
 & =\|\tilde{L}^{-T}\nabla^{2}U(x)\tilde{L}^{-1}\|
\end{align*}
Now we inspect the norm inside the second supremum:
\begin{align*}
\|L\nabla^{2}U(x)^{-1}L^{T}\| & =\|U\Sigma V^{T}\nabla^{2}U(x)^{-1}V^{T}\Sigma U\|\\
 & =\|\Sigma V^{T}\nabla^{2}U(x)^{-1}V^{T}\Sigma\|\\
 & =\|V\Sigma V^{T}\nabla^{2}U(x)^{-1}V^{T}\Sigma V\|\\
 & =\|\tilde{L}\nabla^{2}U(x)^{-1}\tilde{L}^{T}\|
\end{align*}

\subsection{Implications from Assumption \ref{ass:explicit eigenvector}}\label{proof:v_control}

From the statement of Assumption \ref{ass:explicit eigenvector} it is immediate that $v_{i}(x)^Tv_{i}\geq1-\delta$.
Note that the assumption implies the following bound $\|v_{i}(x)-v_{i}\|\leq\sqrt{2}\left(1-\sqrt{1-\delta}\right)$.
For $i,j\in[d]$ such that $i\neq j$, the reverse triangle inequality
gives us that
\begin{align*}
\|v_{i}(x)-v_{j}\| & \geq\|v_{j}-v_{i}\|-\|v_{i}-v_{i}(x)\|\\
 & \geq\sqrt{2}-\sqrt{2}\left(1-\sqrt{1-\delta}\right)
\end{align*}
and so
\[
\sqrt{2(1-\langle v_{i}(x),v_{j}\rangle}\geq\sqrt{2}-\sqrt{2}\left(1-\sqrt{1-\delta}\right)
\]
hence $v_{i}(x)^Tv_{j}\leq\delta$ as required.

\subsection{Proof of Theorem \ref{thm:kappa_L_bound_1}}\label{proof:kappa_L_bound_1}

Perform the eigendecomposition $\nabla^{2}U(x)=O_{x}D_{x}O_{x}^{T}$
for $O_{x}\in O(d)$ with columns $v_{i}(x)$ and $D_{x}\in\mathbb{R}^{d\times d}$
diagonal with elements $\lambda_{i}(x)$. Perform the eigendecomposition $L = V\Sigma V^T$ where $V$ has columns $v_i\in\mathbb{R} ^ d$ for $i \in [d]$ and $\Sigma := \text{diag}\{\sigma_1,...,\sigma_d\}$. Defining $\mathcal{E}_{x}:=V^{T}O_{x}-\mathbf{I}_{d}$,
Assumption \ref{ass:explicit eigenvector} guarantees that the elements of $\mathcal{E}_{x}$ are at most $\delta$
in absolute value. Inspecting the first term in the definition of
$\tilde{\kappa}$, we have that
\begin{align*}
\|L^{-T}\nabla^{2}U(x)L^{-1}\| & =\|\Sigma^{-1}(\mathcal{E}_{x}+\mathbf{I}_{d})D_{x}(\mathcal{E}_{x}+\mathbf{I}_{d})^{T}\Sigma^{-1}\|\\
 & \leq\|\Sigma^{-1}\mathcal{E}_{x}D_{x}\mathcal{E}_{x}^{T}\Sigma^{-1}\|+2\|\Sigma^{-1}\mathcal{E}_{x}D_{x}\Sigma^{-1}\|+\|\Sigma^{-1}D_{x}\Sigma^{-1}\|\\
 & \leq\|\Sigma^{-1}\mathcal{E}_{x}D_{x}\mathcal{E}_{x}^{T}\Sigma^{-1}\|+2\|\Sigma^{-1}\mathcal{E}_{x}D_{x}\Sigma^{-1}\|+(1+\epsilon)
\end{align*}
where the second line is due to the triangle inequality of the matrix
2-norm, the last line due to Assumption \ref{ass:explicit eigenvalue}. Inspecting
the norm in the second term in the above:
\begin{align*}
\|\Sigma^{-1}\mathcal{E}_{x}D_{x}\Sigma^{-1}\|^{2} & =\sup_{\|v\|=1}\sum_{k=1}^{d}\left(\sum_{s=1}^{d}\frac{\lambda_{s}(x)}{\sigma_{s}\sigma_{k}}(\mathcal{E}_{x})_{ks}v_{s}\right)^{2}\\
 & \leq\delta^{2}\sup_{\|v\|=1}\sum_{k=1}^{d}\left(\sum_{s=1}^{d}\frac{\lambda_{s}(x)}{\sigma_{s}\sigma_{k}}v_{s}\right)^{2}\\
 & =\delta^{2}\sum_{k=1}^{d}\frac{1}{\sigma_{k}^{2}}\sup_{\|v\|=1}\left(\sum_{s=1}^{d}\frac{\lambda_{s}(x)}{\sigma_{s}}v_{s}\right)^{2}\\
 & \leq\delta^{2}(1+\epsilon)^{2}\sum_{k=1}^{d}\frac{1}{\sigma_{k}^{2}}\sup_{\|v\|=1}\left(\sum_{s=1}^{d}\sigma_{s}v_{s}\right)^{2}\\
 & =\delta^{2}(1+\epsilon)^{2}\sum_{k=1}^{d}\frac{1}{\sigma_{k}^{2}}\sum_{s=1}^{d}\sigma_{s}^{2}
\end{align*}
where the second line comes from Assumption \ref{ass:explicit eigenvector} and the fourth line comes from Assumption \ref{ass:explicit eigenvalue}. Looking
at the first term now:
\begin{align*}
\|\Sigma^{-1}\mathcal{E}_{x}D_{x}\mathcal{E}_{x}^{T}\Sigma^{-1}\| & =\|\Sigma^{-1}\mathcal{E}_{x}D_{x}^{\frac{1}{2}}\|^{2}\\
 & =\sup_{\|v\|=1}\sum_{k=1}^{d}\left(\sum_{s=1}^{d}\frac{\sqrt{\lambda_{s}(x)}}{\sigma_{k}}(\mathcal{E}_{x})_{ks}v_{s}\right)^{2}\\
 & \leq\delta^{2}\sum_{k=1}^{d}\frac{1}{\sigma_{k}^{2}}\sup_{\|v\|=1}\left(\sum_{s=1}^{d}\sqrt{\lambda_{s}(x)}v_{s}\right)^{2}\\
 & \leq\delta^{2}(1+\epsilon)\sum_{k=1}^{d}\frac{1}{\sigma_{k}^{2}}\sup_{\|v\|=1}\left(\sum_{s=1}^{d}\sigma_{s}v_{s}\right)^{2}\\
 & =\delta^{2}(1+\epsilon)\sum_{k=1}^{d}\frac{1}{\sigma_{k}^{2}}\sum_{s=1}^{d}\sigma_{s}^{2}
\end{align*}
where the third line comes from Assumption \ref{ass:explicit eigenvector} and the fourth line comes from Assumption \ref{ass:explicit eigenvalue}. Putting the terms together yields
\[
\|L^{-T}\nabla^{2}U(x)L^{-1}\|\leq(1+\epsilon)\left(1+\delta\sqrt{\sum_{i=1}^{d}\sigma_{i}^{2}\sum_{i=1}^{d}\sigma_{i}^{-2}}\right)^{2}
\]
Now we follow the same procedure for $\|L\nabla^{2}U(x)L^{T}\|$:
\begin{align*}
\|L\nabla^{2}U(x)L^{T}\| & \leq\|\Sigma\mathcal{E}_{x}D_{x}^{-\frac{1}{2}}\|^{2}+2\|\Sigma\mathcal{E}_{x}D_{x}^{-1}\Sigma\|+\|\Sigma D_{x}^{-1}\Sigma\|\\
 & \leq\|\Sigma\mathcal{E}_{x}D_{x}^{-\frac{1}{2}}\|^{2}+2\|\Sigma\mathcal{E}_{x}D_{x}^{-1}\Sigma\|+(1+\epsilon)
\end{align*}
starting with the second term:
\begin{align*}
\|\Sigma\mathcal{E}_{x}D_{x}^{-1}\Sigma\|^{2} & =\sup_{\|v\|=1}\sum_{k=1}^{d}\left(\sum_{s=1}^{d}\frac{\sigma_{s}\sigma_{k}}{\lambda_{s}(x)}(\mathcal{E}_{x})_{ks}v_{s}\right)^{2}\\
 & \leq\delta^{2}\sum_{k=1}^{d}\sigma_{k}^{2}\sup_{\|v\|=1}\left(\sum_{s=1}^{d}\frac{\sigma_{s}}{\lambda_{s}(x)}v_{s}\right)^{2}\\
 & \leq\delta^{2}(1+\epsilon)^{2}\sum_{k=1}^{d}\sigma_{k}^{2}\sup_{\|v\|=1}\left(\sum_{s=1}^{d}\frac{1}{\sigma_{s}}v_{s}\right)^{2}\\
 & \leq\delta^{2}(1+\epsilon)^{2}\sum_{k=1}^{d}\sigma_{k}^{2}\sum_{s=1}^{d}\frac{1}{\sigma_{s}^{2}}
\end{align*}
and the first term:
\begin{align*}
\|\Sigma\mathcal{E}_{x}D_{x}^{-\frac{1}{2}}\|^{2} & =\sup_{\|v\|=1}\sum_{k=1}^{d}\left(\sum_{s=1}^{d}\frac{\sigma_{k}}{\sqrt{\lambda_{s}(x)}}(\mathcal{E}_{x})_{ks}v_{s}\right)^{2}\\
 & \leq\delta^{2}\sum_{k=1}^{d}\sigma_{k}^{2}\sup_{\|v\|=1}\left(\sum_{s=1}^{d}\frac{1}{\sqrt{\lambda_{s}(x)}}v_{s}\right)^{2}\\
 & \leq\delta^{2}(1+\epsilon)\sum_{k=1}^{d}\sigma_{k}^{2}\sup_{\|v\|=1}\left(\sum_{s=1}^{d}\frac{1}{\sigma_{s}}v_{s}\right)^{2}\\
 & =\delta^{2}(1+\epsilon)\sum_{k=1}^{d}\sigma_{k}^{2}\sum_{s=1}^{d}\frac{1}{\sigma_{s}^{2}}
\end{align*}
from which follows
\[
\tilde{\kappa}\leq(1+\epsilon)^{2}\left(1+\delta\sqrt{\sum_{i=1}^{d}\sigma_{i}^{2}\sum_{i=1}^{d}\sigma_{i}^{-2}}\right)^{4}
\]

\subsection{Tight $\delta$ dependence in Theorem \ref{thm:kappa_L_bound_1}}\label{correct delta dependence}

Take the first norm of (\ref{condition_number}): $\|L^{-T}\nabla^{2}U(x)L^{-1}\|=\|Q_{\pi}GD_{\pi}^{1/2}G^{T}Q_{\pi}^{T}Q_{\pi}D_{\pi}^{-1}Q_{\pi}^{T}Q_{\pi}GD_{\pi}^{1/2}G^{T}Q_{\pi}^{T}\|=\|D_{\pi}^{1/2}G^{T}D_{\pi}^{-1}GD_{\pi}^{1/2}\|$.
Similarly the second norm is $\|D_{\pi}^{-1/2}G^{T}D_{\pi}GD_{\pi}^{-1/2}\|$.
Therefore $\tilde{\kappa}=\lambda_{1}(M)\lambda_{2}(M)^{-1}$ where $M=D_{\pi}^{1/2}G^{T}D_{\pi}^{-1}GD_{\pi}^{1/2}$.
Since $G$ is a perturbation of angle $\arccos(1-\delta)$ it has
the form
\[
G=\begin{pmatrix}1-\delta & -\sqrt{\delta(2-\delta)}\\
\sqrt{\delta(2-\delta)} & 1-\delta
\end{pmatrix}
\]
and so 
\[
M=\begin{pmatrix}(1-\delta)^{2}+\delta(2-\delta)\frac{\lambda_{1}}{\lambda_{2}} & \left(\sqrt{\frac{\lambda_{1}}{\lambda_{2}}}-\sqrt{\frac{\lambda_{2}}{\lambda_{1}}}\right)(1-\delta)\sqrt{\delta(2-\delta)}\\
\left(\sqrt{\frac{\lambda_{1}}{\lambda_{2}}}-\sqrt{\frac{\lambda_{2}}{\lambda_{1}}}\right)(1-\delta)\sqrt{\delta(2-\delta)} & (1-\delta)^{2}+\delta(2-\delta)\frac{\lambda_{2}}{\lambda_{1}}
\end{pmatrix}
\]
where $\lambda_{1}$ and $\lambda_{2}$ are the diagonal elements
of $D_{\pi}$. We have that $\text{tr}(M)=2(1-\delta)^{2}+\delta(2-\delta)(\lambda_{1}\lambda_{2}^{-1}+\lambda_{1}^{-1}\lambda_{2})$
and $\det(M)=1$. That the determinant is one means that $\tilde{\kappa}=\lambda_{1}(M)^{2}$.
Using the trace-determinant formulation of the eigenvalues of a $2\times2$
matrix we have that
\[
\lambda_{1}(M)=\frac{1}{2}\left(2(1-\delta)^{2}+\delta(2-\delta)\left(\frac{\lambda_{1}}{\lambda_{2}}+\frac{\lambda_{2}}{\lambda_{1}}\right)+\sqrt{\left(2(1-\delta)^{2}+\delta(2-\delta)\left(\frac{\lambda_{1}}{\lambda_{2}}+\frac{\lambda_{2}}{\lambda_{1}}\right)\right)^{2}-4}\right)
\]
so, unless $\lambda_{1}=\lambda_{2}$ and $\Sigma_{\pi}$ is a multiple
of the identity, $\tilde{\kappa}=O(\delta^{4})$.

\subsection{Assumption \ref{ass:hessian localisation} implies Assumption \ref{ass:explicit eigenvalue}}\label{proof:A3 implies A1}
Weyl's inequality implies that
\[
\frac{\lambda_{i}(\nabla^{2}U(x))}{\sigma_{i}^{2}}\leq\frac{\lambda_{i}(LL^{T})+\lambda_{1}(\nabla^{2}U(x)-LL^{T})}{\sigma_{i}^{2}}\leq1+\frac{\|\nabla^{2}U(x)-LL^{T}\|}{\sigma_{i}^{2}}
\]
and so $\|\nabla^{2}U(x)-LL^{T}\|\leq\sigma_{d}^{2}\epsilon$ implies
Assumption \ref{ass:explicit eigenvalue} with the same $\epsilon$.

\subsection{Proof of Theorem \ref{thm:kappa_L_bound_2}}\label{proof:kappa_L_bound_2}

Based on the intuition gained from Proposition \ref{prop:symmetry_invariance} we can assume that
$L$ is symmetric, and so its left and right singular vectors are
simply its eigenvectors. Using \citep[Corollary~1]{yu:15}
with $\hat{\Sigma}=\nabla^{2}U(x)$ and $\Sigma=LL^{T}$ we have that
$\|v_{i}(x)-v_{i}\|\leq2^{\frac{3}{2}}\gamma^{-1}\|\nabla^{2}U(x)-LL^{T}\|$.
Rearranging, the Assumption \ref{ass:hessian localisation} gives $\langle v_{i}(x),v_{i}\rangle\geq1-4\gamma^{-2}\sigma_{d}^{-4}\epsilon^{2}$.
From \ref{proof:A3 implies A1}, Assumption \ref{ass:hessian localisation} gives us Assumption \ref{ass:explicit eigenvalue}
with the same $\epsilon$, and hence we can apply Theorem \ref{thm:kappa_L_bound_1}
with $\delta=1-(1-2\gamma^{-1}\sigma_{d}^{-2}\epsilon)^{2}$.

\subsection{Proof of Theorem \ref{thm:kappa_L_bound_3}}\label{proof:kappa_L_bound_3}

Using Proposition \ref{prop:symmetry_invariance} we assume that $L$ is symmetric. For the first supremum in the definition of $\tilde{\kappa}$ note that $\|\nabla^{2}U(x)-L^{2}\|=\|L^{T}(L^{-T}\nabla^{2}U(x)L^{-1}-\mathbf{I}_{d})L\|$.
Using the fact that $\sigma_{i}(BA)\leq\|B\|\sigma_{i}(A)$ and $\sigma_{i}(AC)\leq\sigma_{i}(A)\|C\|$
for matrices $A,B,C$ of appropriate sizes and all $i\in[d]$ \citep[exercise~1.3.24]{tao:12}
we have
\begin{align*}
\|L^{T}(L^{-T}\nabla^{2}U(x)L^{-1}-\mathbf{I}_{d})L\| & \geq\frac{\sigma_{1}(L^{T}(L^{-T}\nabla^{2}U(x)L^{-1}-\mathbf{I}_{d})LL^{-1})}{\sigma_{1}(L^{-1})}\\
 & =\sigma_{d}(L)\|L^{T}(L^{-T}\nabla^{2}U(x)L^{-1}-\mathbf{I}_{d})\|\\
 & \geq\sigma_{d}(L)\frac{\sigma_{1}(L^{-T}L^{T}(L^{-T}\nabla^{2}U(x)L^{-1}-\mathbf{I}_{d}))}{\sigma_{1}(L^{-T})}\\
 & =\sigma_{d}(L)^{2}\|L^{-T}\nabla^{2}U(x)L^{-1}-\mathbf{I}_{d}\|
\end{align*}
Therefore we can bound $\|L^{-T}\nabla^{2}U(x)L^{-1}-\mathbf{I}_{d}\|\leq \epsilon$ using Assumption \ref{ass:hessian localisation}.
Using the reverse triangle inequality $\|L^{-T}\nabla^{2}U(x)L^{-1}-\mathbf{I}_{d}\|\geq|\|L^{-T}\nabla^{2}U(x)L^{-1}\|-1|$
we get $\|L^{-T}\nabla^{2}U(x)L^{-1}\|\leq1+\epsilon$.

For the second supremum in the definition of $\tilde{\kappa}$ we use the same technique as the first supremum,
first noting that $\|\nabla^{2}U(x)^{-1}-L^{-2}\|\leq\|\nabla^{2}U(x)^{-1}\|\|L^{-2}\|\|\nabla^{2}U(x)-L^{2}\|\leq m^{-1}\epsilon$.
Employing the technique from before:
\begin{align*}
\|\nabla^{2}U(x)^{-1}-L^{-2}\| & =\|L^{-1}(L\nabla^{2}U(x)^{-1}L^{T}-\mathbf{I}_{d})L^{-T}\|\\
 & \geq\sigma_{d}(L^{-1})^{2}\|L\nabla^{2}U(x)^{-1}L^{T}-\mathbf{I}_{d}\|
\end{align*}
and hence $\|L\nabla^{2}U(x)^{-1}L^{T}-\mathbf{I}_{d}\|\leq\sigma_1(L)^2m^{-1}\epsilon$.
Using the reverse triangle inequality again gives $\|L\nabla^{2}U(x)^{-1}L^{T}\|\leq1+\sigma_1(L)^2m^{-1}\epsilon$.

\subsection{Proof of Proposition \ref{prop:ostrowski_rectangular}}\label{proof:ostrowski_rectangular}

Applying the non-rectangular form of Ostrowski's theorem \citep[Theorem~3.2]{higham:98a} gives for any $x\in\mathbb{R}^d$
\[
\lambda_d(B^TB)\lambda_{n-d+1}(\Lambda(x)) \leq \lambda_1(B^T \Lambda(x) B) \leq \lambda_1(B^TB)\lambda_1(\Lambda(x)),
\]
and similarly
\[
\lambda_d(B^TB)\lambda_n(\Lambda(x)) \leq \lambda_d(B^T \Lambda(x) B) \leq \lambda_1(B^T B)\lambda_d(\Lambda(x))
\]
Since $\kappa := \sup_{x \in \mathbb{R}^d}\lambda_1(B^T \Lambda(x) B)/\inf_{x \in \mathbb{R}^d}\lambda_d(B^T \Lambda(x) B)$ then applying the upper/lower bound to $\lambda_1(B^T \Lambda(x) B)$ and the lower/upper bound to $\lambda_d(B^T \Lambda(x) B)$ point-wise gives the upper/lower bound on $\kappa$ as desired.

\subsection{Proof of Proposition \ref{prop:mult_dalalyan}}\label{proof:mult_dalalyan}

Setting $\tilde{X}^T = (X^TX)^{-1/2}X^T$ and applying Proposition \ref{prop:ostrowski_rectangular} gives the result, noting that $\tilde{X}^T\tilde{X}=\mathbf{I}_d$.

\subsection{Proof of Proposition \ref{prop:mult_Hessian_at_point}}\label{proof:mult_Hessian_at_point}

First note that the preconditioned Hessian can be written
\[
L^{-T}\nabla^2U(x)L^{-1} = (X^T\Lambda(x^*) X)^{1/2}X^T \Lambda(x^*)^{1/2}\Lambda(x^*)^{-1/2} \Lambda(x) \Lambda(x^*)^{-1/2}\Lambda(x^*)^{1/2} X(X^T\Lambda(x^*) X)^{1/2}
\]
Setting $\tilde{X}^T := (X^T\Lambda(x^*) X)^{1/2}X^T \Lambda(x^*)^{1/2}$ and then applying the upper bound of Proposition \ref{prop:ostrowski_rectangular} to the matrix $\tilde{X}^T \Lambda(x^*)^{-1/2} \Lambda(x) \Lambda(x^*)^{-1/2} \tilde{X}$ gives the first inequality. The second follows from applying the same bound again to $\Lambda(x^*)^{-1/2} \Lambda(x) \Lambda(x^*)^{-1/2}$ and noting that
\[
\kappa(\Lambda(x^*))\leq\frac{\sup_{x\in\mathbb{R}^d}\lambda_1(\Lambda(x))}{\inf_{x\in\mathbb{R}^d}\lambda_d(\Lambda(x))}
\]

\subsection{Proof of Theorem \ref{RWM_mixing_bounds}}\label{proof:RWM_mixing_bounds}

If we take $\sigma^{2}=\xi / (Md)$, then \citep[Theorem~1]{andrieu:22}
implies that the spectral gap $\gamma_{\kappa}$ of the RWM algorithm
on a target with a $m$-strongly convex, $M$-smooth potential is bounded as follows:
\[
C\xi\exp(-2\xi)\frac{1}{\kappa}\frac{1}{d}\leq\gamma_{\kappa}\leq\frac{\xi}{2}\frac{1}{d}
\]
We will modify the proof of \citep[Lemma 47]{andrieu:22} so that the upper
bound on $\gamma_{\kappa}$ subsequently depends on $\kappa$. The
spectral gap of the RWM algorithm on a target $\pi$ with kernel $P$
is defined as $\gamma_{k}:=\inf_{f\in L_{0}^{2}(\pi)}(\mathcal{E}(P,f)/\text{Var}_{\pi}(f))$
where $\mathcal{E}$ is the Dirichlet form associated with $\pi$.
Define $g(x):=\langle v_{\text{max}},x-\mathbb{E}_{\pi}[X]\rangle$ where $v_{\text{max}}\in\mathbb{R}^{d}$
is the eigenvector associated with the greatest eigenvalue of $E_{\pi}\left(\nabla^{2}U(X)\right)^{-1}$.
The Cramér-Rao inequality gives that
\begin{align*}
\text{Var}_{\pi}(g(X)) & \geq v_{\text{max}}^{T}\mathbb{E}_{\pi}\left[\nabla^{2}U(X)\right]^{-1}v_{\text{max}}\\
 & =\lambda_{1}\left(\mathbb{E}_{\pi}\left[\nabla^{2}U(X)\right]^{-1}\right)\|v_{\text{max}}\|^{2}\\
 & =\lambda_{d}\left(\mathbb{E}_{\pi}\left[\nabla^{2}U(X)\right]\right)^{-1}\|v_{\text{max}}\|^{2}
\end{align*}
The second equality comes from the fact that $\mathbb{E}_{\pi}\left[\nabla^{2}U(X)\right]^{-1}$
is positive definite, since it is the inverse of the expectation of
a matrix that is itself positive definite.

Say $v \in \mathbb{R}^{d}$ is the eigenvector associated with the
smallest eigenvalue $\lambda_{d}(y)$ of $\nabla^2U(y)$ for a given $y\in\mathbb{R}^d$.
Then
\begin{align*}
\lambda_{d}\left(\mathbb{E}_{\pi}\left[\nabla^{2}U(X)\right]\right) & =\inf_{\|v\|=1}v^{T}\mathbb{E}_{\pi}\left[\nabla^{2}U(X)\right]v\\
 & \leq v^{T}\mathbb{E}_{\pi}\left[\nabla^{2}U(X)\right]v\\
 & =\mathbb{E}_{\pi}\left[v^{T}\left(\nabla^{2}U(X)-\nabla^2U(y)+\nabla^2U(y)\right)v\right]\\
 & \leq\sup_{x\in\mathbb{R}^{d}}\|\nabla^{2}U(x)-\nabla^2U(y)\|+\lambda_d(y)\\
 & \leq m(1+2\epsilon)
\end{align*}
where in the final line we use Assumption \ref{ass:Hessian only localisation}, and the fact that
Assumption \ref{ass:Hessian only localisation} implies
\[
\frac{\lambda_d(y)}{\lambda_d(x)}\leq 1 + \epsilon
\]
for all $x,y\in\mathbb{R}^d$ (see \ref{proof:A3 implies A1}) and hence $\lambda_d(y)\leq(1+\epsilon)\lambda_{d}(x)\leq(1+\epsilon)m$.
Therefore $\text{Var}_{\pi}(g(X))\geq m^{-1}(1+2\epsilon)^{-1}\|v_{\text{max}}\|^{2}$.
Upper bounding the Dirichlet form in the same way as \citep[Lemma~47]{andrieu:22}
gives $\mathcal{E}(P,g)\leq(1/2)\sigma^{2}\|v_{\text{max}}\|^{2}$,
and so
\[
\gamma_{k}=\inf_{f\in L_{0}^{2}(\pi)}\frac{\mathcal{E}(P,f)}{\text{Var}_{\pi}(f)}\leq\frac{\mathcal{E}(P,g)}{\text{Var}_{\pi}(g)}\leq\frac{\frac{1}{2}\sigma^{2}\|v_{\text{max}}\|^{2}}{m^{-1}(1+2\epsilon)^{-1}\|v_{\text{max}}\|^{2}}=\frac{1}{2}\xi\kappa^{-1}d^{-1}(1+2\epsilon)
\]

\subsection{Proof of Corollary \ref{cor:improved_gap}}\label{proof:improved_gap}

The lower bound for the spectral gap post-preconditioning is $\gamma_{\tilde{\kappa}}\geq C\xi\exp(-2\xi)\tilde{\kappa}^{-1}d^{-1}$
due to \citep[Theorem~1]{andrieu:22}. The target satisfies Assumption \ref{ass:hessian localisation}
so we can use Theorem \ref{thm:kappa_L_bound_3} to modify the bound: $\gamma_{\tilde{\kappa}}\geq C\xi\exp(-2\xi)(1+\epsilon)^{-1}\left(1+m^{-1}\sigma_1(L)^2\epsilon\right)^{-1}d^{-1}$. So, applying the upper bound found in Theorem \ref{RWM_mixing_bounds} to
the spectral gap before preconditioning, we see that a condition number
$\kappa$ such that
\[
\frac{1}{2}\xi\kappa^{-1}d^{-1}(1+2\epsilon')\leq C\xi\exp(-2\xi)(1+\epsilon)^{-1}\left(1+\frac{\sigma_1(L)^2}{m}\epsilon\right)^{-1}d^{-1}
\]
guarantees that $\gamma_{\tilde{\kappa}}\geq\gamma_{\kappa}$ and so increases
the spectral gap.

\subsection{Proof of Proposition \ref{prop:optimal_OU}}\label{proof:optimal_OU}
The preconditioned O-U process has generator $A_{L}=\text{tr}\left(L^{-1}L^{-T}\nabla^{2}\right)+\left\langle -L^{-1}L^{-T}\Sigma_{\pi}^{-1}x,\nabla\right\rangle $
for $x\in\mathbb{R}^{d}$, where $\nabla$ is the grad operator and
\[
\nabla^{2}:=\sum_{i=1}^{d}\frac{\partial^{2}}{\partial x_{i}^{2}}
\]
The spectral gap of a generator of a stochastic process with generator $A$ is defined as the smallest distance from its spectrum (excluding 0) to 0.

Denoting spectrum of $-L^{-1}L^{-T}\Sigma_{\pi}^{-1}$ by $\lambda_{1},...,\lambda_{d}$ \citet[Theorem~3.1]{metafune:02} has that the spectrum of the generator
$A_{L}$ is $\{\lambda=\sum_{i=1}^{d}n_{i}\lambda_{i}:n_{i}\in\mathbb{N}\}$.
Therefore the spectral gap is 
\[
\min\left\{ \left|\sum_{i=1}^{d}n_{i}\lambda_{i}\right|:n_{i}\in\mathbb{N}\backslash\{0\}\text{ for }i\in [d]\right\} =\min_{i\in[d]}\left|\lambda_{i}\right|
\]
Say $\left|\lambda_{i}\right|>1$
for some $i$ without loss of generality. Then the constraint $\left|\det(-L^{-1}L^{-T}\Sigma_{\pi}^{-1})\right|=1$
implies that $\left|\lambda_{j}\right|<1$
for some $j\neq i$, and hence the resulting spectral gap is strictly
less than one. In the case that $\left|\lambda_{i}\right|=1$
for all $i$, the spectral gap is exactly one, and therefore optimal.
This is achieved when $L=\Sigma_{\pi}^{-1/2}$ for any notion
of the matrix square root.

\subsection{Proof of Proposition \ref{prop:localise_covariance}}\label{proof:localise_covariance}

We assume WLOG that $U(x^{*})=0$. Taylor's theorem with integral
remainder has that
\[
U(x)=\int_{0}^{1}(1-t)(x-x^{*})^{T}\nabla^{2}U(x^{*}+t(x-x^{*}))(x-x^{*})dt
\]
(since $U(x^{*})=0$, $\nabla U(x^{*})=0$) from which we can deduce
\[
\frac{1}{2}(x-x^{*})^{T}\Delta_{-}(x-x^{*})\leq U(x)\leq\frac{1}{2}(x-x^{*})^{T}\Delta_{+}(x-x^{*})
\]
and hence 
\[
\exp\left(-\frac{1}{2}(x-x^{*})^{T}\Delta_{+}(x-x^{*})\right)\leq\exp(-U(x))\leq\exp\left(-\frac{1}{2}(x-x^{*})^{T}\Delta_{-}(x-x^{*})\right)
\]

with 
\[
\frac{Z_{\Delta_{+}}}{Z}\frac{1}{Z_{\Delta_{+}}}\exp\left(-\frac{1}{2}(x-x^{*})^{T}\Delta_{+}(x-x^{*})\right)\leq\frac{1}{Z}\exp(-U(x))\leq\frac{Z_{\Delta_{-}}}{Z}\frac{1}{Z_{\Delta_{-}}}\exp\left(-\frac{1}{2}(x-x^{*})^{T}\Delta_{-}(x-x^{*})\right)
\]
where $Z_{A}:=\sqrt{(2\pi)^{d}\det A^{-1}}$. For an arbitrary
$v\in\mathbb{R}^{d}$ we have that
\begin{align*}
v^{T}\Sigma_{\pi}v & =\frac{1}{Z}\int(v^{T}(x-\mu_{\pi}))^{2}\exp(-U(x))dx\\
 & \leq\frac{Z_{\Delta_{-}}}{Z}\mathbb{E}_{\mathcal{N}(x^{*},(\Delta_{-})^{-1})}[(v^{T}(X-\mu_{\pi}))^{2}]\\
 & \leq\frac{Z_{\Delta_{-}}}{Z}\left(\mathbb{E}_{\mathcal{N}(x^{*},(\Delta_{-})^{-1})}[(v^{T}(X-x^{*}))^{2}]-(v^{T}(x^{*}-\mu_{\pi}))^{2}\right)\\
 & \leq\frac{Z_{\Delta_{-}}}{Z_{\Delta_{+}}}\left(v^{T}(\Delta_{-})^{-1}v-v^{T}(x^{*}-\mu_{\pi})(x^{*}-\mu_{\pi})^{T}v\right)
\end{align*}
where the last inequality follows from the fact that $Z_{\Delta_{+}}\leq Z\leq Z_{\Delta_{-}}$.
We can construct a similar lower bound to give $$c((\Delta_{+})^{-1}-(x^{*}-\mu_{\pi})(x^{*}-\mu_{\pi})^{T})\leq\Sigma_{\pi}\leq c^{-1}((\Delta_{-})^{-1}-(x^{*}-\mu_{\pi})(x^{*}-\mu_{\pi})^{T})$$
defining $c:=(Z_{\Delta_{+}}/Z_{\Delta_{-}})=\sqrt{\det\Delta_{-}\det\Delta_{+}^{-1}}\leq1$.
This gives $P_{-}\leq\Sigma_{\pi}^{-1}\leq P_{+}$ where 
\begin{align*}
P_{+} & :=c^{-1}\left((\Delta_{+})^{-1}-(x^{*}-\mu_{\pi})(x^{*}-\mu_{\pi})^{T}\right)^{-1}\\
P^{-} & :=c\left((\Delta_{-})^{-1}-(x^{*}-\mu_{\pi})(x^{*}-\mu_{\pi})^{T}\right)^{-1}
\end{align*}
and hence
\begin{align*}
P_{+} & =c^{-1}\left(\Delta_{+}+\left(1-(x^{*}-\mu_{\pi})^{T}\Delta_{+}(x^{*}-\mu_{\pi})\right)^{-1}\Delta_{+}(x^{*}-\mu_{\pi})(x^{*}-\mu_{\pi})^{T}\Delta_{+}\right)\\
P_{-} & =c\left(\Delta_{-}+\left(1-(x^{*}-\mu_{\pi})^{T}\Delta_{-}(x^{*}-\mu_{\pi})\right)^{-1}\Delta_{-}(x^{*}-\mu_{\pi})(x^{*}-\mu_{\pi})^{T}\Delta_{-}\right)
\end{align*}
using the Woodbury identity. The fact that $(x^{*}-\mu_{\pi})^{T}\Delta_{\pm}(x^{*}-\mu_{\pi}) = \text{Tr}(D_\pm)$ gives the result.

We have that $\|\nabla^{2}U(x)-\Sigma_{\pi}^{-1}\|:=\sup_{v}\left|v^{T}\nabla^{2}U(x)v-v^{T}\Sigma_{\pi}^{-1}v\right|$.
Say the quantity inside the absolute value is positive. Then $v^{T}\nabla^{2}U(x)v-v^{T}\Sigma_{\pi}^{-1}v\leq v^{T}\nabla^{2}U(x)v-v^{T}P_{-}v$.
Now say the quantity is negative, giving us $v^{T}\Sigma_{\pi}^{-1}v-v^{T}\nabla^{2}U(x)v\leq v^{T}P_{+}v-v^{T}\nabla^{2}U(x)v$.
In sum this gives
\begin{align*}
\|\nabla^{2}U(x)-\Sigma_{\pi}^{-1}\| & \leq\sup_{v:\|v\|=1}\max\left\{ v^{T}\nabla^{2}U(x)v-v^{T}P_{-}v,v^{T}P_{+}v-v^{T}\nabla^{2}U(x)v\right\} \\
 & \leq\sup_{v:\|v\|=1}\max\left\{ v^{T}\Delta_{+}v-v^{T}P_{-}v,v^{T}P_{+}v-v^{T}\Delta_{-}v\right\} \\
 & \leq\max\left\{\|\Delta_{+}-P_{-}\|,\|P_{+}-\Delta_{-}\|\right\}
\end{align*}

\subsection{Proof of Proposition \ref{cor:localise_covariance_additive}}\label{proof:localise_covariance_additive}

From proposition \ref{prop:localise_covariance} we have that $P_{-}\leq\Sigma_{\pi}^{-1}\leq P_{+}$where
\begin{align*}
P_{+} & =c^{-1}\left(A+\epsilon\mathbf{I}_{d}+\left(1-(x^{*}-\mu_{\pi})^{T}(A+\epsilon\mathbf{I}_{d})(x^{*}-\mu_{\pi})\right)^{-1}(A+\epsilon\mathbf{I}_{d})(x^{*}-\mu_{\pi})(x^{*}-\mu_{\pi})^{T}(A+\epsilon\mathbf{I}_{d})\right)\\
P_{-} & =c\left(A-\epsilon\mathbf{I}_{d}+\left(1-(x^{*}-\mu_{\pi})^{T}(A-\epsilon\mathbf{I}_{d})(x^{*}-\mu_{\pi})\right)^{-1}(A-\epsilon\mathbf{I}_{d})(x^{*}-\mu_{\pi})(x^{*}-\mu_{\pi})^{T}(A-\epsilon\mathbf{I}_{d})\right)
\end{align*}
since $\Delta_{-}=A-\epsilon\mathbf{I}_{d}$ and $\Delta_{+}=A+\epsilon\mathbf{I}_{d}$.
The bounds stated at the end of the proposition give
\begin{align*}
\|\nabla^{2}U(x)-\Sigma_{\pi}^{-1}\| & \leq\max\left\{ \|\Delta_{+}-P_{-}\|,\|P_{+}-\Delta_{-}\|\right\} \\
 & =\max\left\{ \|(1-c)A+(1+c)\epsilon\mathbf{I}_{d}-c\tilde{P}_{-}\|,\|(c^{-1}-1)A+(c^{-1}+1)\epsilon\mathbf{I}_{d}-c^{-1}\tilde{P}_{+}\|\right\} \\
 & \leq(c^{-1}-1)\|A\|+(c^{-1}+1)\epsilon+\max\left\{c\|\tilde{P}_{-}\|,c^{-1}\|\tilde{P}_{+}\|\right\}
\end{align*}
where in the final line we use the triangle inequality.

\section{}\label{appendix:B}

The full matrix from equation (\ref{gaussian_covariance}) in Section \ref{subsubsec:diagonal_counterproductive} is as follows:

\begin{equation*}
    \Sigma_{\pi}=
\begin{pmatrix}21.548973 & 5.678587 & 18.667787 & 4.463119 & 6.855300\\
\star & 2.028958 & 4.863393 & 1.208146 & 2.109502\\
\star & \star & 16.261735 & 3.926604 & 5.726388\\
\star & \star & \star & 1.405213 & 1.409477\\
\star & \star & \star & \star & 2.905902
\end{pmatrix}
\end{equation*}

\bibliographystyle{unsrtnat}






\end{document}